\documentclass[11pt]{amsart}
\usepackage{amssymb,mathrsfs,graphicx,enumerate}
\usepackage{amsmath,amsfonts,amssymb,amscd,amsthm,bbm}
\usepackage{tikz}
\usetikzlibrary{matrix}
\usepackage[retainorgcmds]{IEEEtrantools}
\usepackage{colortbl}

\title[An infinite-dimensional Schr\"{o}dinger-Lohe hierarchy]{On the Schr\"{o}dinger-Lohe hierarchy for aggregation  and its emergent dynamics}

\author[Ha]{Seung-Yeal Ha}
\address[Seung-Yeal Ha]{\newline Department of Mathematical Sciences\newline Seoul National University, Seoul 08826 and \newline
Korea Institue for Advanced Study, Hoegiro 85, Seoul 02455, Republic of Korea}
\email{syha@snu.ac.kr}

\author[Park]{Hansol Park}
\address[Hansol Park]{\newline Department of Mathematical Sciences\newline Seoul National University, Seoul 08826, Republic of Korea}
\email{hansol960612@snu.ac.kr}

\newtheorem{theorem}{Theorem}[section]
\newtheorem{lemma}{Lemma}[section]

\newtheorem{proposition}{Proposition}[section]
\newtheorem{remark}{Remark}[section]

\newcommand{\bbr}{\mathbb R}

\newcommand{\bbz}{\mathbb Z}

\newcommand{\bbt} {\mathbb T}

\newcommand{\bbc}{\mathbb C}

\begin{document}

\date{\today}


\thanks{\textbf{Acknowledgment.} The work of S.-Y.Ha is supported by NRF-2020R1A2C3A01003881}

\begin{abstract}
The Lohe hierarchy is a hierarchy of finite-dimensional aggregation models consisting of the Kuramoto model, the complex Lohe sphere model, the Lohe matrix model and the Lohe tensor model. In contrast, the Schr\"{o}dinger-Lohe model is the only known infinite-dimensional Lohe aggregation model in literature. In this paper, we provide an explicit connection between the Schr\"{o}dinger-Lohe model and the complex Lohe sphere model, and then by exploiting this explicit relation, we construct infinite-dimensional liftings of the Lohe matrix and the Lohe tensor models. In this way, we establish the Schr\"{o}dinger-Lohe hierarchy which corresponds to the infinite-dimensional extensions of the Lohe hierarchy. For the proposed hierarchy, we provide sufficient frameworks leading to the complete aggregation in terms of coupling strengths and initial configurations.
\end{abstract}

\maketitle \centerline{\date}


\section{Introduction} \label{sec:1}
\setcounter{equation}{0}
Collective behaviors often appear in classical and quantum many-body systems, e.g., aggregation of bacteria, herding of sheep, schooling of fish, synchronous firing of fireflies and array of Josephson junctions in semiconductors \cite{A-B, A-B-F, D-B1, H-K-P-Z, P-R, St, VZ} etc. Despite of their ubiquity in our nature, model-based studies on the collective dynamics were first begun only in a half century ago by two pioneers, Arthur Winfree and Yoshiki Kuramoto \cite{Ku1, Ku2, Wi2, Wi1}.Recently,  due to applications in the control of drones, self-driving cars and sensor networks, research on the collective dynamics has received lots of attentions from diverse scientific disciplines such as applied mathematics, biology, control theory and statistical physics, etc.

In this paper, we are interested in the Kuramoto model \cite{C-H-J-K, C-S, H-L-X} and its high-dimensional extensions such as  the Lohe sphere model \cite{C-C-H,C-H5, H-K-R, M-T-G}, the Lohe matrix model \cite{H-R, Lo-0, Lo-1, Lo-2}. See \cite{B-C-S,D-F-M-T,D-F-M, De, M-T-G, T-M, Zhu} for other related models. Aforementioned aggregation models were further extended to the ensemble of Lohe tensors by the authors in \cite{H-P3} in which we call it as {\it the Lohe tensor model} which completes the Lohe hierarchy (LH) comprising of finite-dimensional aggregation models such as the Kuramoto model, the Lohe sphere(LS) model, the Lohe matrix(LM) model and the Lohe tensor(LT) models, whereas in an infinite-dimensional setting, the Schr\"{o}dinger-Lohe(SL) model is the only known Lohe type aggregation model so far. In what follows, we address the following two questions: \newline
\begin{itemize}
\item
(Q1):~What is the connection between the SL model and the finite-dimensional aggregation models in the LH? 

\vspace{0.2cm}

\item
(Q2):~If such connection exists, can we establish a Schr\"{o}dinger-Lohe hierarchy(SLH) consisting of the infinite-dimensional analogs of the aggregation models in LH?
\end{itemize}
\vspace{0.2cm}

The main results of this paper are positive answers for the above two questions. First, we provide an explicit connection between the SL model and the LS model. If  the solution to the SL model is expanded in terms of a basis consisting of suitable standing wave solutions, we show that coefficients satisfy the Lohe sphere model on $(L^2 \cap L^{\infty})(\bbz_+)$. 
Second, we employ the idea of connecting the SL model and LS model to introduce infinite-dimensional analogs of the LM and LT models which will be coined as the SL matrix and SL tensor models. Since the details are rather messy, we will not go into the details here and we instead leave the detailed results in Section \ref{sec:4} and Section \ref{sec:5}. In this manner, we establish the Schr\"{o}dinger-Lohe hierarchy. \newline

The rest of this paper is organized as follows. In Section \ref{sec:2}, we review minimal materials on tensors which are enough to understand the rest of paper, and then briefly review the Lohe hierarchy consisting of the Kuramoto model, complex Lohe sphere model, the generalized Lohe matrix model and the Lohe tensor model, and study basic properties such as conservation law and solution splitting property of each model. In Section \ref{sec:3}, we study a priori estimates on the SL model and propose a new SL type model with rotational couplings and then discuss its connection with the Kuramoto mdoel.
In Section \ref{sec:4}, we present an explit bridge between the Schr\"{o}dinger Lohe model and the complex Lohe sphere model on $(\ell^2 \cap \ell^{\infty})(\bbz_+)$, and then using this idea of an explicit bridge, we provide an extension of the generalized Lohe matrix model to a Schr\"{o}dinger setting. In Section \ref{sec:5}, we further propose a Schr\"{o}dinger type extension of the Lohe tensor model and finally establish the Schr\"{o}dinger-Lohe hierarchy. In some sense, unlike to the standard approaches such as classical and quantum BBGKY hierarchies in classical and quantum physics which is a top-down approach, whereas our approach is a bottom-up approach to go from low-rank models to high-rank models. Finally, Section \ref{sec:6} is devoted to a brief summary of main results and some remaining issues to be discussed in a future work.

\vspace{0.5cm}

\noindent {\bf Notation}: For complex-valued functions $\psi$ and $\varphi$ in $L^2(\bbt^d)$, we define the inner product and its associated norm as follows:
\[ \langle \psi | \varphi \rangle  = \int_{\mathbb{T}^d} \overline{\psi(x)} \varphi(x) dx, \quad \| \psi \|_{2} :=  \sqrt{\langle\psi|\psi\rangle}. \] 

\section{Preliminaries} \label{sec:2}
\setcounter{equation}{0}
In this section, we briefly review basics material of tensors, tensor space and tensor contraction, and then we introduce the Lohe hierarchy and we review their basic properties such as conservation laws and solution splitting property.

\subsection{Tensors and tensor contraction} \label{sec:2.1}
A {\it tensor} denotes a multi-dimensional array of complex numbers with several indices. Thus, it can be viewed as a generalization of vector and matrix, and the rank of a tensor is the number of indices, i.e., a rank-$m$ tensor of dimensions $d_1\times \cdots\times  d_m$ is an element of ${\mathbb C}^{d_1 \times \cdots \times d_m}$. Hence, a rank-$m$ tensor $T \in \bbc^{d_1 \times \cdots \times d_m}$  can also be identified as a multilinear map from $\bbc^{d_1} \times \cdots \bbc^{d_m}$ to $\bbc$. Complex numbers, complex vectors and complex matrices correspond to rank-0, 1 and 2 tensors, respectively.  

For a rank-$m$ tensor $T$ and a multi-index $\alpha_* = (\alpha_1, \cdots, \alpha_m) \in \{1, \cdots, d_1 \} \times \cdots \times \{1, \cdots, d_m \}$, we denote the $\alpha_*$-th component of $T$ by $[T]_{\alpha_*} = [T]_{\alpha_1 \cdots \alpha_m}$, and we also denote ${\bar T}$ by the rank-$m$ tensor whose components are the complex conjugate of the elements of $T$:
\[ [{\bar T}]_{\alpha_1 \cdots \alpha_m} =\overline{[T]_{\alpha_1 \cdots \alpha_m}}. \]
Finally, we set ${\mathcal T}_m(\bbc): = {\mathcal T}_m(d_1, \cdots, d_m;{\mathbb C})$ to be the set of all rank-$m$ tensors with complex entries and the size $d_1 \times \cdots \times d_m$. One of key basic operations in ${\mathcal T}_m({\mathbb C})$ is a tensor contraction which yields a low-rank tensor by contracting repeated variables in the expressions. Note that the inner product between rank-1 tensors  and matrix product between rank-2 tensors can be defined as special cases of tensor contractions: for $v, w \in {\mathcal T}_1(d_1; \bbc)$ and $A, B  \in {\mathcal T}_2(d_1, d_1;\bbc)$, 
\[ \langle v, w \rangle := [{\bar v}]_\alpha [w]_\alpha,  \qquad  [AB]_{\alpha \beta} :=  [A]_{\alpha \gamma} [B]_{\gamma \beta}. \] 
where we used Einstein summation rule for repeated indices. \newline

\noindent For a rank-$m$ tensor $T \in {\mathcal T}_m(\bbc)$, we also set 
\begin{align*}
\begin{aligned}
& [T]_{\alpha_{*}}:=[T]_{\alpha_{1}\alpha_{2}\cdots\alpha_{m}},\quad [T]_{\alpha_{*0}}:=[T]_{\alpha_{10}\alpha_{20}\cdots\alpha_{m0}}, \quad [T]_{\alpha_{*1}}:=[T]_{\alpha_{11}\alpha_{21}\cdots\alpha_{m1}}, \\
&[T]_{\alpha_{*i_*}}:=[T]_{\alpha_{1i_1}\alpha_{2i_2}\cdots\alpha_{mi_m}}, \qquad  [T]_{\alpha_{*(1-i_*)}}:=[T]_{\alpha_{1(1-i_1)}\alpha_{2(1-i_2)}\cdots\alpha_{m(1-i_m)}}.
\end{aligned}
\end{align*}
Moreover, for a special rank-$2m$ tensor $S \in {\mathcal T}_{2m}(d_1, \cdots, d_m, d_1, \cdots, d_m; \bbc)$, one has 
\[ [S]_{\alpha_*\beta_*}:=[S]_{\alpha_{1}\alpha_{2}\cdots\alpha_{m}\beta_1\beta_2\cdots\beta_{m}}. \]
Next, we define Frobenius inner product, corresponding norm on ${\mathcal T}_m(\bbc)$, ensemble diameter as follows: for a tensor ensemble $\{T_i \} \subset {\mathcal T}_m(\bbc)$, 
\[  \langle T_i, T_j \rangle_F:= [\bar{T}_i]_{\alpha_{*0}} [T_j]_{\alpha_{*0}}, \quad \| T_i \|_F^2 :=  \langle T_i, T_i \rangle_F  \quad \mbox{and} \quad {\mathcal D}(T) := \max_{1 \leq i,j \leq N} \|T_i - T_j \|_F . \]
For an elementary introduction to tensors and elementary tensor operations, we refer to introductory articles \cite{B-C, Or}.

\subsection{The Lohe hierarchy} \label{sec:2.2}
In this subsection, we review the Lohe hierarchy consisting of finite-dimensional Lohe type aggregation models and their basic properties such as conservation laws and solution splitting property:
\[ \mbox{Lohe tensor model}~~\Longrightarrow~~\mbox{Lohe matrix model}~~\Longrightarrow~~\mbox{Lohe sphere model}. \]
\subsubsection{The Lohe tensor model} \label{sec:2.2.1}
Let $\{T_j \}$ be a homogeneous Lohe tensor flock whose dynamics is governed by the following Cauchy problem: 
\begin{align}
\begin{aligned}  \label{B-1}
& \frac{d}{dt}[T_j]_{\alpha_{*0}} =[F]_{\alpha_{*0}\beta_*}[T_j]_{\beta_*} \\
& \hspace{1.5cm} +\sum_{i_* \in \{0,1\}^m}{\kappa_{i_*}}([T_c]_{\alpha_{*i_*}}[\bar{T}_j]_{\alpha_{*1}}[T_j]_{\alpha_{*(1-i_*)}}-[T_j]_{\alpha_{*i_*}}[\bar{T}_c]_{\alpha_{*1}}[T_j]_{\alpha_{*(1-i_*)}}),  \\
& T_j \Big|_{t = 0} = T_j^{in}, \quad j = 1, \cdots, N,
\end{aligned}
\end{align}
where $\kappa_{i_*}$ is a nonnegative coupling strength and $F$ is a skew-hermitian rank-$2m$ tensor in ${\mathcal T}_{2m}(d_1, \cdots, d_m, d_1, \cdots, d_m; \bbc)$ with the following three properties: for $n \in {\mathbb Z}_+$,
\begin{align}
\begin{aligned}   \label{B-2}
&  [F]_{\alpha_*\beta_*}=-[\bar{F_j}]_{\beta_*\alpha_*}, \quad  [F^0]_{\alpha_*\beta_*}=\delta_{\alpha_*\beta_*}, \\
& [F^n]_{\alpha_*\beta_*}=[F]_{\alpha_*\gamma_{1*}}[F]_{\gamma_{1*}\gamma_{2*}}\cdots[F]_{\gamma_{(n-1)*}\beta_{*}}, \quad   [FT]_{\alpha_*}=[F]_{\alpha_*\beta_*}[T]_{\beta_*}, \\
& \delta_{\alpha_{*0}\gamma_{*0}}\delta_{\gamma_{*1}\alpha_{*1}}=[e^{-Ft}]_{\alpha_{*0}\beta_{*0}}[e^{Ft}]_{\beta_{*i_*}\gamma_{*i_*}}[e^{-{F}t}]_{\alpha_{*1}\beta_{*1}}[e^{Ft}]_{\beta_{*(1-i_*)}\gamma_{*(1-i_*)}}, 
\end{aligned}
\end{align}
for all $i_* \in \{0, 1\}^m$ with $\kappa_{i_*} \neq 0$ and $\delta_{\alpha_* \beta_*}$ is defined as follows.
\[   \delta_{\alpha_* \beta_*} = \begin{cases}
1, \quad & \alpha_k = \beta_k\quad\forall k=1, 2, \cdots, m, \\
0, \quad & \mbox{otherwise.}
\end{cases} \]
Although \eqref{B-1} and \eqref{B-2} look so complicate, it admits a conservation law and solution splitting property. For this, we consider the Cauchy problem to the subsystem of \eqref{B-1} with zero free flow $F \equiv 0$ and the same initial data:
\begin{align}
\begin{aligned}  \label{B-3}
& \frac{d}{dt} [S_j]_{\alpha_{*0}}
= \sum_{i_*}\kappa_{i_*} \Big([S_c]_{\alpha_{*i_*}} [\bar{S}_j]_{\alpha_{*1}}[S_j]_{\alpha_{*(1-i_*)}}
-[S_j]_{\alpha_{*i_*}} [\bar{S}_c]_{\alpha_{*1}}[S_j]_{\alpha_{*(1-i_*)}}\Big),\\
& S_j \Big|_{t = 0} = T_j^{in}, \quad j=1, 2, \cdots, N.
\end{aligned}
\end{align}

\vspace{0.2cm}

\begin{proposition} \label{P2.1}
\emph{\cite{H-P0, H-P3}}
Let $\{T_j \}$ and $\{S_j \}$ be solutions to \eqref{B-1} - \eqref{B-2} and \eqref{B-3}  with the same initial data $\{T_j^{in} \}$, respectively. Then, the following assertions hold.
\begin{enumerate}
\item
$\|T_j \|_F$ is a conserved quantity:
\[ \|T_j(t) \|_{F} = \|T_j^{in} \|_F, \quad t \geq 0,~~j = 1, \cdots, N. \]
\item
The Lohe tensor flow \eqref{B-2} can be represented as a composition of free flow and nonlinear flow.
\[  T_j(t) = e^{tF} S_j(t), \quad t \geq 0, \quad j = 1, \cdots, N,  \]
where $e^{tF}$ is a matrix exponential defined by the following relation:
\[  [e^{tF}]_{\alpha_*\beta_*}=\sum_{n=0}^{\infty} \frac{t^n}{n!} [F^n]_{\alpha_*\beta_*}. \]
\end{enumerate}
\end{proposition}
\subsubsection{The Lohe matrix model} \label{sec:2.2.2}
 Let $\{A_j \}$ be a collection of complex $d_1 \times d_2$ Lohe matrices whose dynamics is governed by the following Cauchy problem \cite{H-P1}:
\begin{align}
\begin{aligned}\label{B-4}
& {\dot A}_j= B A_j+\kappa_0(A_c A_j^\dagger A_j-A_j A_c^\dagger A_j)+\kappa_1(A_jA_j^\dagger A_c-A_jA_c^\dagger A_j),\quad t>0,\\
& A_j \Big|_{t = 0} =A_j^{in}, \quad j=1, 2, \cdots, N,
\end{aligned}
\end{align}
where $\kappa_0$ and $\kappa_1$ are nonnegative coupling strengths, and $\dagger$ denotes the Hermitian conjugate,  $A_c := \frac{1}{N} \sum_{k=1}^{N} A_k$ and the rank-4 tensor $B \in {\mathcal T}_4(d_1, d_2, d_1, d_2; \bbc)$ satisfies
\begin{align}
\begin{aligned} \label{B-5}
&  [{\bar B}]_{\alpha \beta \gamma \delta} = -[B]_{\gamma \delta \alpha \beta}, \quad 1 \leq \alpha, \gamma \leq d_1,~~1 \leq \beta, \delta \leq d_2, \quad j = 1, \cdots, N, \\
& [e^{-Bt}]_{\alpha\beta\gamma\delta}[e^{Bt}]_{\gamma\epsilon\alpha_1\beta_1}[e^{-Bt}]_{\alpha_2\beta_2\psi\epsilon}[e^{Bt}]_{\psi\delta\alpha_3\beta_3}=\delta_{\alpha_1\alpha}\delta_{\beta_3\beta}\delta_{\beta_1\beta_2}\delta_{\alpha_2\alpha_3}. 
\end{aligned}
\end{align}
Consider the corresponding nonlinear sub-system associated with \eqref{B-4}:
\begin{align}
\begin{aligned} \label{B-6}
& {\dot N}_j= \kappa_1(N_c N_j^\dagger N_j - N_j N_c^\dagger N_j)+\kappa_2( N_j N_j^\dagger N_c- N_j N_c^\dagger N_j),  \quad t > 0, \\
& N_j \Big|_{t = 0} = A_j^{in}.
\end{aligned}
\end{align}
Similar to Proposition \ref{P2.1}, we have a conservation law and solution splitting property. 
\begin{proposition} \label{P2.2}
\emph{\cite{H-P2}}
Let $\{A_j \}$ and $\{N_j \}$ be solutions to \eqref{B-4} - \eqref{B-5} and \eqref{B-6}, respectively. Then, the following assertions hold.
\begin{enumerate}
\item
$\|A_j \|_F$ is a conserved quantity:
\[ \|A_j(t) \|_{F} = \|A_j^{in} \|_F, \quad t \geq 0,~~j = 1, \cdots, N. \]
\item
The Lohe matrix flow can be represented as a composition of free flow and nonlinear flow.
\[  A_j(t) = e^{tB} N_j(t), \quad t \geq 0, \quad j = 1, \cdots, N,   \]
where $e^{tB}$ is given as follows.
\[  [e^{tB}]_{\alpha\beta\gamma\delta} :=\sum_{n=0}^\infty \frac{t^n}{n!}[B^n]_{\alpha\beta\gamma\delta}.   \]
\end{enumerate}
\end{proposition}
\subsubsection{The complex Lohe sphere model} \label{sec:2.2.3}
Let $\{v_j \}$ be a collection of the complex vectors in $\bbc^d$ whose dynamics is governed by the following Cauchy problem:
\begin{align}
\begin{aligned}\label{B-7}
& \dot{v}_j=\Omega v_j+\kappa_0\big(v_c \langle v_j, v_j \rangle -\langle{v_c, v_j}\rangle v_j\big)+\kappa_1\big(\langle{v_j, v_c}\rangle - \langle{v_c, v_j}\rangle\big)v_j,\qquad t\geq0,\\
& v_j \Big|_{t = 0} =v_j^{in}, \quad j=1, 2, \cdots, N,
\end{aligned}
\end{align}
where $\kappa_0$ and $\kappa_1$ are nonnegative coupling strengths, $v_c := \frac{1}{N} \sum_{i=1}^{N} v_i$ and $\Omega$ is $d \times d$ skew-Hermitian with the property $\Omega^\dagger = -\Omega.$ \newline

\noindent Note that for a real vector $v_j = x_j \in \bbr^d$, the third term  in the R.H.S. of \eqref{B-7} vanishes, and system \eqref{B-7} reduces to the complex Lohe sphere model in \cite{C-C-H}:
\[ 
\dot{x}_j= \Omega x_j +\kappa_{0} \Big(x_c\langle x_j, x_j \rangle-x_j \langle x_c, x_j \rangle \Big).
\]
Next, we consider the corresponding nonlinear subsystem associated with \eqref{B-7}:
\begin{align}
\begin{aligned} \label{B-8}
& \dot{w}_j=  \kappa_{0} (w_c\langle w_j, w_j \rangle-w_j \langle w_c, w_j \rangle)+\kappa_1 (\langle w_j, w_c \rangle- \langle w_c,  w_j \rangle) w_j, \quad t > 0, \\
& w_j \Big|_{t = 0} = v_j^{in}, \quad j = 1, \cdots, N.
\end{aligned}
\end{align}
Similar to Proposition \ref{P2.1}, we have a conservation law and  solution splitting property. 
\begin{proposition} \label{P2.3}
\emph{\cite{H-P2}}
Let $\{v_j \}$ and $\{w_j \}$ be solutions to \eqref{B-7} and \eqref{B-8} with the same initial data $\{v_j^{in} \}$, respectively. Then, the following assertions hold.
\begin{enumerate}
\item
$\|v_j\|$ is a conserved quanity:
\[ \|v_j(t)\| = \|v_j^{in} \|, \quad t \geq 0,~~j = 1, \cdots, N. \]
\item
The Lohe sphere flow can be represented as a composition of free flow and nonlinear flow:
\[ v_j=e^{Bt} w_j, \quad j = 1, \cdots, N.    \]
\end{enumerate}
\end{proposition}

\section{Schr\"{o}dinger-Lohe type models} \label{sec:3}
\setcounter{equation}{0} 
In this section, we first  review the Schr\"{o}dinger-Lohe model and its basic properties, and then  we introduce a variant of the Schr\"{o}dinger-Lohe model with rotational couplings and study its connection with the Kuramoto model. 
\subsection{The Schr\"{o}dinger-Lohe model}  \label{sec:3.1}
Let $\{ \psi_j \}$ be a collection of $N$ complex-valued functions in ${\mathcal C}(\bbr_+; L^2(\bbt^d))$ whose dynamics is governed by the following Cauchy problem:
\begin{align}
\begin{aligned} \label{C-1}
& \mathrm{i}\partial_t \psi_j= {\mathcal H} \psi_j + \frac{\mathrm{i}\kappa}{N}\sum_{k=1}^N (\psi_k \langle \psi_j | \psi_j \rangle -\langle{\psi_k| \psi_j}\rangle \psi_j),\qquad (t, x)\in \bbr_+ \times  \bbt^d, \\
& \psi_j \Big|_{t = 0} =\psi_j^{in},\quad j = 1, \cdots, N,
\end{aligned}
\end{align}
where ${\mathcal H}= -\frac{1}{2} \Delta_{x} + V(x)$ is a one-body Hermitian Hamiltonian. The global existence of strong and smooth solutions to \eqref{C-1} was studied in \cite{H-H} using the standard energy method and asymptotic dynamics of \eqref{C-1} has been extensively discussed in \cite{H-H-K1, H-H-K2}.  Now, we consider the corresponding nonlinear flow for \eqref{C-1} with the same initial data:
\begin{align}
\begin{aligned} \label{C-2}
& \mathrm{i}\partial_t \varphi_j= \frac{\mathrm{i}\kappa}{N}\sum_{k=1}^N (\varphi_k \langle \varphi_j | \varphi_j \rangle -\langle{\varphi_k| \varphi_j}\rangle \varphi_j),\quad (t, x)\in \bbr_+ \times  \bbt^d, \\
& \varphi_j \Big|_{t = 0} =\psi_j^{in},\quad j = 1, \cdots, N,
\end{aligned}
\end{align}
Next, we present the conservation of $L^2$-norm and the solution splitting property. 
\begin{proposition} \label{P3.1}
\emph{\cite{C-H4}}
Let $\{\psi_j \}$ and $\{\varphi_j \}$ be global smooth solutions to \eqref{C-1} and \eqref{C-2}, respectively. Then, the following assertions hold.
\begin{enumerate}
\item
$\| \psi_j \|_{2}$ is a conserved quantity:
\[  \| \psi_j(t) \|_2 = \|\psi_j^{in} \|_2, \quad t \geq 0,~~j = 1, \cdots, N. \]
\item
The Schr\"{o}dinger-Lohe flow can be represented as a composition of free flow and nonlinear flow:
\[  \psi_j = e^{-\mathrm{i}{\mathcal H}t} \varphi_j, \quad j = 1, \cdots, N. \]
\end{enumerate}
\end{proposition}
\subsection{The  Schr\"{o}dinger-Lohe model with rotational couplings} \label{sec:3.2} Below, we introduce a variant of the SL model motivated. Recall that two coupling terms in  the complex Lohe sphere model:
\begin{equation} \label{C-5-1}
 \kappa_0\big(v_c-\langle{v_c, v_j}\rangle v_j\big)+\kappa_1\big(\langle{v_j, v_c}\rangle - \langle{v_c, v_j}\rangle\big)v_j.
 \end{equation}
The coupling term involving with $\kappa_0$ has the same structure as that of the SL model discussed in the previous subsection. Next, we propose a SL type model motivated by the coupling term involving with $\kappa_1$ responsible for rotational motion. \newline

Consider the Cauchy problem to the Schr\"{o}dinger-Lohe model with rotational couplings:
\begin{align}
\begin{aligned} \label{C-6}
& \mathrm{i} \partial_t \psi_j ={\mathcal H} \psi_j +\frac{\mathrm{i}\kappa}{N}\sum_{k=1}^N (\langle \psi_j | \psi_k \rangle -\langle \psi_k| \psi_j \rangle ) \psi_j, \quad (t,x)\in \bbr_+ \times \bbt^d, \\
& \psi_j \Big|_{t = 0} =\psi_j^{in}, \quad j = 1, \cdots, N.
\end{aligned}
\end{align}
Note that the coupling term can be rewritten using the average wave function $\psi_c := \frac{1}{N} \sum_{j=1}^{N} \psi_j$:
\[ \frac{\mathrm{i}\kappa}{N}\sum_{k=1}^N (\langle \psi_j | \psi_k \rangle -\langle \psi_k| \psi_j \rangle ) \psi_j =  \langle \psi_j| \psi_c \rangle - \langle \psi_c | \psi_j \rangle.    \]
This exactly coincides with the second term in \eqref{C-5-1}. The global well-posedness of \eqref{C-6} can be treated similarly as in \cite{H-H}. Thus, we focus on the a priori asymptotic dynamics for \eqref{C-6}. \newline

Next, we study a connection between \eqref{C-6} and the Kuramoto model. For this, we consider the following setting:
\begin{equation} \label{C-7}
{\mathcal H} \equiv 0, \qquad \psi_j^{in}(x)=e^{\mathrm{i}\theta^{in}_j}\psi(x),
\end{equation}
where $\psi$ is an $L^2$-function with $\| \psi \|_{2} = 1$.  
\begin{proposition} \label{P3.2}
Suppose that the setting \eqref{C-7} holds, and let $\{\psi_j \}$ be a global smooth solution to \eqref{C-6}. Then, one has
\begin{equation*}
\begin{cases} \label{C-8}
\displaystyle \psi_j(t,x) =e^{\mathrm{i}\theta_j(t)}\psi(x), \quad (t, x) \in \bbr_+ \times \bbt^d, \\
\displaystyle {\dot \theta}_j =\frac{2\kappa}{N}\sum_{k=1}^N\sin(\theta_k-\theta_j), \quad j = 1, \cdots, N, \\
 \theta_j \Big|_{t = 0} =\theta_j^{in}.
 \end{cases}
\end{equation*}
\end{proposition}
\begin{proof}
Note that $\psi_j$ satisfies 
\begin{align}
\begin{aligned} \label{C-8-1}
 \partial_t\psi_j &=\frac{\kappa}{N}\sum_{k=1}^N (\langle{\psi_j| \psi_k}\rangle -\langle{\psi_k| \psi_j}\rangle )\psi_j = \kappa(\langle{\psi_j|\psi_c}\rangle-\langle{\psi_c|\psi_j }\rangle)\psi_j  \\
 &= 2\kappa\mathrm{i}\cdot\mathrm{Im}(\langle{\psi_j|\psi_c}\rangle)\psi_j,
\end{aligned}
\end{align}
where we used $\overline{\langle \psi_c|\psi_j \rangle} =   \langle \psi_j|\psi_c \rangle$. \newline

\noindent This yields
\begin{equation*} \label{C-9}
\psi_j(x, t)=\psi_j^{in}(x) e^{2\kappa\mathrm{i} \int_0^t \mathrm{Im}(\langle{\psi_j |\psi_c}\rangle)(s)ds}= e^{2\kappa\mathrm{i} \int_0^t \mathrm{Im}(\langle{\psi_j |\psi_c}\rangle)(s)ds} e^{\mathrm{i}\theta_j^{in}} \psi(x).
\end{equation*}
Thus, it is reasonable to set the ansatz for $\psi_j$ as follows.
\begin{equation} \label{C-9-1}
\psi_j(t,x) =e^{\mathrm{i}\theta_j(t)}\psi(x).   
\end{equation}
This implies 
\[ \partial_t \psi_j = {\mathrm i} {\dot \theta}_j  \psi_j,  \quad \langle \psi_j | \psi_k \rangle = e^{{\mathrm i} (\theta_k - \theta_j)} \| \psi \|_{L^2}^2 = e^{{\mathrm i} (\theta_k - \theta_j)}. \]
We use the above calculation and the ansatz \eqref{C-9-1} into \eqref{C-8-1} to see
\[
{\mathrm i} {\dot \theta}_j  \psi_j = \frac{\kappa}{N}\sum_{k=1}^N\left(e^{\mathrm{i}(\theta_k-\theta_j)}-e^{\mathrm{i}(\theta_j-\theta_k)}\right)\psi_j
\]
which yields
\[ {\dot \theta}_j =\frac{2\kappa}{N}\sum_{k=1}^N\sin(\theta_k-\theta_j).
\]
\end{proof}
Next, we show that system \eqref{C-7} admits a conservation law as the original S-L model. 
\begin{lemma}\label{L3.1}
Let $\{\psi_j \}$ be a global smooth solution to \eqref{C-6}. Then, for $t > 0$ and $i, j = 1, \cdots, N$, one has
\[ \frac{d}{dt}\langle{\psi_i|\psi_j}\rangle= \kappa (\langle{\psi_c|\psi_i-\psi_j}\rangle-\langle{\psi_i-\psi_j|\psi_c}\rangle )\langle{\psi_i|\psi_j}\rangle, \qquad \|\psi_j(t) \|_{2} = \| \psi_j^{in} \|_{2}. \]
\end{lemma}
\begin{proof} (i)~We use \eqref{C-6} to get 
\begin{align}
\begin{aligned} \label{C-10}
\frac{d}{dt}\langle{\psi_i|\psi_j}\rangle &=\langle{\partial_t\psi_i|\psi_j}\rangle+\langle{\psi_i|\partial_t\psi_j}\rangle =\langle{-\mathrm{i}{\mathcal H} \psi_i|\psi_j}\rangle+\langle{\psi_i|-\mathrm{i} {\mathcal H} \psi_j}\rangle \\
&+\frac{\kappa}{N}\sum_{k=1}^N \Big(\langle{\psi_k|\psi_i}\rangle-\langle{\psi_i|\psi_k}\rangle+\langle{\psi_j| \psi_k}\rangle -\langle{\psi_k| \psi_j}\rangle \Big)\langle{\psi_i|\psi_j}\rangle\\
&=\frac{\kappa}{N}\sum_{k=1}^N(\langle{\psi_k|\psi_i-\psi_j}\rangle-\langle{\psi_i-\psi_j|\psi_k}\rangle )\langle{\psi_i|\psi_j}\rangle.
\end{aligned}
\end{align}
(ii)~We set $i = j$ in \eqref{C-10} to see
\[ \frac{d}{dt}\langle{\psi_j |\psi_j}\rangle = \frac{d}{dt} \| \psi_j \|_{2}^2 = 0. \]
\end{proof}
\noindent For a global smooth solution $\{\psi_j \}$ and $i_1, i_2, \cdots, i_m\in\{1, 2, \cdots, N\}$, we introduce a functional $ {\mathcal J}_{i_1 \cdots i_m} (\Psi)$:
\begin{equation} \label{C-11}
 {\mathcal J}_{i_1 \cdots i_m} (\Psi) := \langle{\psi_{i_1}|\psi_{i_2}}\rangle\cdot\langle{\psi_{i_2}|\psi_{i_3}}\rangle\cdot\cdots\cdot\langle{\psi_{i_m}|\psi_{i_1}}\rangle.  
\end{equation} 
\begin{proposition} \label{P3.3}
Let $\{\psi_i\}$ be be a global smooth solution to \eqref{C-6}. Then, we have two conservation laws:
\[
\frac{d}{dt}{\mathcal J}_{i_1 \cdots i_m} (\Psi) =0 \quad \mbox{and} \quad  \frac{d}{dt}|\langle{\psi_i|\psi_j}\rangle|^2=0, \quad  t > 0.
\]
\end{proposition}
\begin{proof} (i)~We use \eqref{C-10} and \eqref{C-11} to obtain
\begin{align*}
\begin{aligned}
&\frac{d}{dt}\big(\langle{\psi_{i_1}|\psi_{i_2}}\rangle\cdot\langle{\psi_{i_2}|\psi_{i_3}}\rangle\cdot\cdots\cdot\langle{\psi_{i_m}|\psi_{i_1}}\rangle\big) =\langle{\psi_{i_1}|\psi_{i_2}}\rangle\cdot\langle{\psi_{i_2}|\psi_{i_3}}\rangle\cdot\cdots\cdot\langle{\psi_{i_m}|\psi_{i_1}}\rangle \\
& \hspace{1cm} \times \sum_{k=1}^N \Big(\langle \psi_k|(\psi_{i_1}-\psi_{i_2})+(\psi_{i_2}-\psi_{i_3})+\cdots+(\psi_{i_m}-\psi_{i_1})\rangle\\
& \hspace{1.2cm}-\langle(\psi_{i_1}-\psi_{i_2})+(\psi_{i_2}-\psi_{i_3})+\cdots+(\psi_{i_m}-\psi_{i_1})|\psi_k\rangle \Big)=0.
\end{aligned}
\end{align*}
(ii)~We set
\[ m=2, \quad i_1=i,\quad i_2=j \]
to get the desired estimate. 
\end{proof}
Next, we recall Barbalat's lemma to be used crucially in the following sections. 
\begin{lemma}  \label{L3.2}
\emph{\cite{Ba}}
\emph{(i)}~Suppose that  a real-valued function $f: [0, \infty) \to \bbr$ is uniformly continuous and it satisfies
\[ \lim_{t \to \infty} \int_0^t f(s)d s \quad \textup{exists}. \]
Then, $f$ tends to zero as $t \to \infty$:
\[ \lim_{t \to \infty} f(t) = 0. \]
\emph{(ii)}~Suppose that a real-valued function $f: [0, \infty) \to \bbr$ is continuously differentiable, and $\lim_{t \to \infty} f(t) = f_\infty \in \bbr$. If $f^{\prime}$ is uniformly continuous, then 
\[ \lim_{t \to \infty} f^{\prime}(t)  = 0. \]
\end{lemma}
\begin{theorem} \label{T3.1}
Let $\{\psi_i\}$ be a global smooth solution to \eqref{C-6} with $\kappa > 0$.  Then, we have 
\begin{eqnarray*}
&& (i)~\frac{d}{dt} \sum_{i, j=1}^N \|\psi_i - \psi_j \|_{2}^2 =-2N\kappa\sum_{i=1}^N|\langle{\psi_c|\psi_i}\rangle-\langle{\psi_i|\psi_c}\rangle|^2\leq0. \\
&& (ii)~\lim_{t \to \infty} \sum_{i=1}^N|\langle{\psi_c|\psi_i}\rangle-\langle{\psi_i|\psi_c}\rangle|^2 = 0.
\end{eqnarray*}
\end{theorem}
\begin{proof}
\noindent (i)~We use the conservation of $\|\psi_i \|_{2} = 1$ to get 
\begin{align}
\begin{aligned} \label{NNN-1}
&\frac{d}{dt} \sum_{i, j=1}^N \|\psi_i-\psi_j \|^2_{2}=-\frac{d}{dt}\sum_{i, j=1}^N\left(\langle{\psi_i|\psi_j}\rangle+\langle{\psi_j|\psi_i}\rangle\right)\\
& \hspace{1cm} =-\frac{\kappa}{N}\sum_{i, j, k=1}^N(\langle{\psi_k|\psi_i-\psi_j}\rangle-\langle{\psi_i-\psi_j|\psi_k}\rangle)(\langle{\psi_i|\psi_j}\rangle-\langle{\psi_j|\psi_i}\rangle)\\
& \hspace{1cm}  =-\frac{2\kappa}{N}\sum_{i, j, k=1}^N(\langle{\psi_k|\psi_i}\rangle-\langle{\psi_i|\psi_k}\rangle)(\langle{\psi_i|\psi_j}\rangle-\langle{\psi_j|\psi_i}\rangle)\\
& \hspace{1cm}  =-2N\kappa\sum_{i=1}^N(\langle{\psi_c|\psi_i}\rangle-\langle{\psi_i|\psi_c}\rangle)(\langle{\psi_i|\psi_c}\rangle-\langle{\psi_c|\psi_i}\rangle)\\
& \hspace{1cm}  =-2N\kappa\sum_{i=1}^N|\langle{\psi_c|\psi_i}\rangle-\langle{\psi_i|\psi_c}\rangle|^2\leq0.
\end{aligned}
\end{align}
\noindent (ii)~It follows from the result of (i) and boundedness that  $\sum_{i, j=1}^N \|\psi_i-\psi_j \|^2_{2}$
converges as time goes infinity. Then, it follows from the boundedness of $\frac{d \psi_i}{dt}$ that  
\[
\frac{d^2}{dt^2}\sum_{i, j=1}^N \|\psi_i-\psi_j \|^2_{2} =-2N\kappa\frac{d}{dt}\left(\sum_{i=1}^N(\langle{\psi_c|\psi_i}\rangle-\langle{\psi_i|\psi_c}\rangle)(\langle{\psi_i|\psi_c}\rangle-\langle{\psi_c|\psi_i}\rangle)\right).
\]
By Barbalat's lemma, one has
\[
\lim_{t\rightarrow\infty}\frac{d}{dt} \sum_{i, j=1}^N \|\psi_i-\psi_j \|^2_{2}=0
\]
This and \eqref{NNN-1} yield
\[
\lim_{t \to \infty} \sum_{i=1}^N|\langle{\psi_c|\psi_i}\rangle-\langle{\psi_i|\psi_c}\rangle|^2 = 0.
\]
\end{proof}
\begin{remark} By the non-increasing property of the relative $L^2$-distances between wave functions, we can see that system \eqref{C-6} does not admit a periodic solution except equilibrium solutions.
\end{remark}
Next, we study the solution splitting property of \eqref{C-6}.  Consider the corresponding nonlinear system:
\begin{align}
\begin{aligned} \label{C-12}
& \mathrm{i} \partial_t \varphi_j = \frac{\mathrm{i}\kappa}{N}\sum_{k=1}^N (\langle \varphi_j | \varphi_k \rangle -\langle \varphi_k| \varphi_j \rangle ) \varphi_j, \quad (t,x)\in \bbr_+ \times \bbt^d, \\
& \varphi_j \Big|_{t = 0} =\psi_j^{in}, \quad j = 1, \cdots, N.
\end{aligned}
\end{align}
\begin{theorem} \label{T3.2}
Let $\{\psi_j\}$ and $\{\varphi_j \}$ be global smooth solutions to \eqref{C-6} and \eqref{C-12}, respectively. Then, one has
\[  \psi_j(t, x) = e^{-\mathrm{i}{\mathcal H}t} \varphi_j(t,x), \quad j = 1, \cdots, N. \]
\end{theorem}
\begin{proof}
It follows from \eqref{C-6} that 
\[
\mathrm{i}\partial_t\psi_i={\mathcal H} \psi_i+\frac{\mathrm{i}\kappa}{N}\sum_{k=1}^N (\langle{\psi_i| \psi_k}\rangle -\langle{\psi_k| \psi_i}\rangle )\psi_i.
\]
Then, we have
\[
\mathrm{i}\partial_t(e^{\mathrm{i}{\mathcal H} t}\psi_i)=\frac{\mathrm{i}\kappa}{N}\sum_{k=1}^N (\langle{e^{\mathrm{i}{\mathcal H} t}\psi_i| e^{\mathrm{i}{\mathcal H} t}\psi_k}\rangle -\langle{e^{\mathrm{i}{\mathcal H} t}\psi_k| e^{\mathrm{i}{\mathcal H} t}\psi_i}\rangle )e^{\mathrm{i}{\mathcal H} t}\psi_i.
\]
This yields the desired result.
\end{proof}

In the following two sections, we introduce two new models in the Schr\"{o}dinger-Lohe hierarchy.

\section{The Schr\"{o}dinger-Lohe matrix model} \label{sec:4}
\setcounter{equation}{0}
In this section, we briefly discuss basic properties to the SL model and present the infinite-dimensional analog of the complex LS model.

\subsection{A bridge between the SL and LS models} \label{sec:4.1}
For a given ${\mathcal H} = -\frac{1}{2} \Delta_x + V(x)$,  let $\{\phi_{\alpha_1}\}$ and $\{ E_{\alpha_1} \}$ be an orthonormal basis consisting of  eigenfunctions and their corresponding eigenvalues for ${\mathcal H}$: 
\[
{\mathcal H} \phi_{\alpha_1} =E_{\alpha_1} \phi_{\alpha_1}, \quad \alpha_1 = 1, 2, \cdots.
\]
Then the standing wave solution $\Phi_{\alpha_1}(t,x) :=e^{-\mathrm{i}E_{\alpha_1} t}\phi_{\alpha_1}(x)$ satisfies the linear Schr\"{o}dinger equation:
\[
\mathrm{i}\partial_t \Phi_{\alpha_1} = {\mathcal H} \Phi_{\alpha_1}, \quad \alpha_1  = 1, 2,\cdots,
\]
and we set $\psi_j$ to be a linear combination of $\{\Phi_{\alpha_1}\}_{\alpha_1}$ as follows:
\begin{equation} \label{D-0}
\psi_j(t,x) =\sum_{\alpha_1}[v_{j}(t)]_{\alpha_1} \Phi_{\alpha_1}(t,x), \quad j = 1, \cdots, N.
\end{equation}
Suppose that $\psi_j$ satisfies the SL model with $\|\psi_j \|_2 = 1$:
\begin{equation} \label{D-0-1}
\mathrm{i}\partial_t \psi_j= {\mathcal H} \psi_j +\frac{\mathrm{i}\kappa}{N}\sum_{k=1}^N (\psi_k-\langle{\psi_k| \psi_j}\rangle \psi_j).
\end{equation}
We use \eqref{D-0} to rewrite the L.H.S. of \eqref{D-0-1} to see 
\begin{align}
\begin{aligned} \label{D-1}
\mathrm{i}\partial_t\psi_j &=\sum_{\alpha_1} \left( [v_{j}]_{\alpha_1} \mathrm{i}\partial_t\Phi_{\alpha_1} +[\dot{v}_{j}]_{\alpha_1} \mathrm{i}\Phi_{\alpha_1} \right) =\sum_{\alpha_1} \left( [v_{j}]_{\alpha_1} {\mathcal H} \Phi_{\alpha_1} + [\dot{v}_{j}]_{\alpha_1} \mathrm{i}\Phi_{\alpha_1} \right)  \\
&={\mathcal H} \psi_j+ {\mathrm i} \sum_{\alpha_1} [\dot{v}_{j}]_{\alpha_1} \Phi_{\alpha_1}.
\end{aligned}
\end{align}
Now, we equate \eqref{D-0-1} and \eqref{D-1} to get 
\begin{align*}
\begin{aligned}
&{\mathcal H} \psi_j+\mathrm{i}\sum_{\alpha}[\dot{v}_{j}]_{\alpha_1}\Phi_{\alpha_1} = {\mathcal H} \psi_j+\frac{\mathrm{i}\kappa}{N}\sum_{k=1}^N(\psi_k-\langle{\psi_k|\psi_j}\rangle \psi_j) \\
& \hspace{1cm} = {\mathcal H} \psi_j+\frac{\mathrm{i}\kappa}{N}\sum_{k=1}^N\sum_{\alpha_1} ( [v_{k}]_{\alpha_1}-\langle{\psi_k|\psi_j}\rangle [v_{j}]_{\alpha_1})\Phi_{\alpha_1}.
\end{aligned}
\end{align*}
This yields
\[
\sum_{\alpha_1} [\dot{v}_{j}]_{\alpha_1}\Phi_{\alpha_1} =\frac{\kappa}{N}\sum_{k=1}^N\sum_{\alpha_1} ( [v_{k}]_{\alpha_1}-\langle{\psi_k|\psi_i}\rangle [v_{j}]_{\alpha_1})\Phi_{\alpha_1}.
\]
Since $\{\Phi_{\alpha_1} \}$ is an orthonormal basis, one has 
\begin{equation}\label{D-2}
\frac{d}{dt} [v_{j}]_{\alpha_1}=\frac{\kappa}{N}\sum_{k=1}^N([v_{k}]_{\alpha_1}-\langle{\psi_k|\psi_j}\rangle [v_{j}]_{\alpha_1}), \quad j = 1, \cdots, N,~~ \alpha_1 = 1, 2, \cdots.
\end{equation}
For each $j = 1, \cdots, N$, we define an infinite complex vector in $(\ell^\infty \cap \ell^2)(\bbz_+)$:
\[   v_j = ([v_{j}]_{1}, [v_{j}]_{2}, \cdots ). \]
We use the definition of $\langle \cdot | \cdot \rangle$ to get
\begin{align} \label{D-3}
\begin{aligned}
\langle{\psi_k | \psi_j}\rangle &=\sum_{\alpha_1, \beta_1} \Big \langle [v_{k}]_{\alpha_1}\Phi_{\alpha_1} \Big | [v_{j}]_{\beta_1}\Phi_{\beta_1} \Big \rangle=\sum_{\alpha_1, \beta_1}[\bar{v}_{k}]_{\alpha_1} [v_{j}]_{\beta_1} \Big \langle \Phi_{\alpha_1} \Big | \Phi_{\beta_1} \Big \rangle \\
&=\sum_{\alpha_1} [\bar{v}_{k}]_{\alpha_1} [v_{j}]_{\alpha_1} = \langle v_k | v_j \rangle.
\end{aligned}
\end{align}
Finally, we combine \eqref{D-2} and \eqref{D-3} to derive an infinite-dimensional counterpart for the complex Lohe sphere model on $(\ell^2 \cap \ell^\infty)(\bbz_+)$:
\begin{equation}\label{D-4}
{\dot v}_j =\frac{\kappa}{N}\sum_{k=1}^N(v_k-\langle{v_k|v_j}\rangle v_j), \quad j = 1, \cdots, N.
\end{equation}
%

\subsection{The SLM model} \label{sec:4.2}
In the previous subsection, we showed that the SL model can be reduced to the extended complex LS model on $(\ell^2 \cap \ell^{\infty})(\bbz_+)$. Below, we propose a Schrodinger-Lohe type model which can be reduced to the LM model in Section \ref{sec:2.2}, and study its emergent dynamics in a priori setting by assuming a global well-posedness of a smooth solution.  \newline

Recall the complex LM model for $d_1 \times d_2$ complex matrix $A_j$:
\begin{equation} \label{D-4-1}
\underbrace{{\dot A}_j - B A_j}_{\mbox{free flow}} = \underbrace{\kappa_0(A_c A_j^*A_j-A_j A_c^*A_j)+\kappa_1(A_jA_j^*A_c-A_jA_c^*A_j)}_{\mbox{cubic mean-field interactions}}.
\end{equation}
In the sequel, we present the Schr\"{o}dinger-Lohe type model which can be associated with the generalized Lohe matrix model \eqref{D-4-1}. 
\newline

Next, we introduce the Schrodinger-Lohe matrix model(SLM) for a homogeneous ensemble. First, we set 
\[ {\mathcal H}  :=-\frac{1}{2} \Delta_{x_1}-\frac{1}{2} \Delta_{x_2} +V, \quad V = V(x_1, x_2), \quad (x_1, x_2) \in \bbt^{d} \times \bbt^d,
\]
where the one-body potential is assumed to be continuous for our emergent dynamics. However, for a global well-posedness of classical solutions based on energy method, we might need to assume high regularity of the potential. \newline

For notational simplicity, we suppress $t$-dependence on $\Psi$ and use a handy notation for a partial inner product:
\[ \Psi_j(x_1, x_2) \equiv \Psi_j(t, x_1, x_2), \quad t \geq 0,~~(x_1, x_2) \in \bbt^d \times \bbt^d, \]
and for $\Psi(x_1, x_2)$ and ${\tilde \Psi}(x_1, x_2)$, we set 
\begin{align}
\begin{aligned} \label{NNN-2}
& \langle \Psi(x_2^*)| {\tilde \Psi}(x_2) \rangle := \int_{\bbt^d} \overline{\Psi(x_1^*, x_2^*)} {\tilde \Psi}(x^*_1,  x_2) dx_1^*, \\
& \langle \Psi(x_1^*)| {\tilde \Psi}(x_1) \rangle := \int_{\bbt^d} \overline{\Psi(x_1^*, x_2^*)} {\tilde \Psi}(x_1,  x^*_2) dx_2^*.
\end{aligned}
\end{align}
Now, we propose the Schrodinger-Lohe matrix(SLM) model as follows: for $t > 0$ and $x_i, x_i^* \in \bbt^d$, 
\begin{equation} \label{D-5}
\begin{cases}
\displaystyle \mathrm{i}\partial_t \Psi_j(x_1, x_2) - {\mathcal H} \Psi_j(x_1, x_2) \\
\displaystyle = \mathrm{i}\kappa_0 \int_{\bbt^{2d}} \Big(  \Psi_c(x_1, x_2^*)     \overline{ \Psi_j(x_1^*, x_2^*)}\Psi_j(x_1^*, x_2) 
 - \Psi_j(x_1, x_2^*)  \overline{\Psi_c(x_1^*, x_2^*)}\Psi_j(x_1^*, x_2)  \Big)dx_1^* dx_2^*\\
\displaystyle + \mathrm{i}\kappa_1 \int_{\bbt^{2d}} \Big( \Psi_j(x_1, x_2^*)     \overline{ \Psi_j(x_1^*, x_2^*)}\Psi_c(x_1^*, x_2)-\Psi_j(x_1, x_2^*)     \overline{ \Psi_c(x_1^*, x_2^*)}\Psi_j(x_1^*, x_2) \Big) dx_1^*dx_2^*, \\
\displaystyle  \Psi_j \Big|_{t = 0} = \Psi_j^{in}, \quad j = 1, \cdots, N,
\end{cases}
\end{equation}
where $\Psi_c := \frac{1}{N} \sum_{k=1}^{N} \Psi_k$. \newline

Under the handy notation \eqref{NNN-2}, system \eqref{D-5} becomes  
\begin{equation} \label{D-5-1}
\begin{cases}
\displaystyle \mathrm{i}\partial_t \Psi_j(x_1, x_2) - {\mathcal H} \Psi_j(x_1, x_2) \\
\displaystyle \hspace{0.2cm} = \mathrm{i}\kappa_0 \int_{\bbt^d} \Big( \Big \langle \Psi_j(x_2^*) \Big | \Psi_j(x_2)  \Big \rangle \Psi_c(x_1, x_2^*)   
 - \Big \langle \Psi_c(x_2^*)  \Big | \Psi_j(x_2) \Big \rangle \Psi_j(x_1, x_2^*)  \Big) dx_2^*\\
\displaystyle  \hspace{0.4cm} + \mathrm{i}\kappa_1 \int_{\bbt^d} \Big(  \Big \langle \Psi_j(x_2^*) \Big | \Psi_c(x_2) \Big \rangle  - \Big \langle \Psi_c(x_2^*) \Big| \Psi_j(x_2) \Big \rangle \Big) \Psi_j(x_1, x_2^*) dx_2^*, \\
\displaystyle  \Psi_j \Big|_{t = 0} = \Psi_j^{in}, \quad j = 1, \cdots, N.
\end{cases}
\end{equation}
Since a global well-posedness of \eqref{D-5} can be treated using a standard energy method as in \cite{H-H} for the SL model in a suitable Sobolev space setting, we will focus on the emergent dynamics in a priori setting.  Notice that the R.H.S. of \eqref{D-4-1} and \eqref{D-5} are structurally the same. \begin{proposition} \label{P4.1}
Let $\{\Psi_j \}$ be a global smooth solution to \eqref{D-5} with the initial data $\|\Psi^{in}_j \| = 1$.  Then $L^2$-norm of $\Psi_j$ is a conserved quantity:
\[
\frac{d}{dt}\|\Psi_j(t)\|_{2}=0, \quad t > 0,~~i = 1, \cdots, N.
\]
where 
\[  \|\Psi_j(t) \|_{2} := \int_{\bbt^{2d}} |\Psi_j(t, x_1, x_2) |^2 dx_2 dx_1, \quad t \geq 0. \]
\end{proposition}
\begin{proof} By definition of $\|\Psi_j\|_{2}^2$, one has
\begin{equation} \label{D-5-0-0}
\frac{d}{dt}\|\Psi_j\|_{2}^2 =\int_{\bbt^{2d}} (\partial_t \Psi_j(x_1, x_2) )\overline{\Psi_j(x_1, x_2)}dx_1 dx_2 + \mbox{(c.c.)},
\end{equation}
where (c.c.) denotes the complex conjugate of the first term. \newline

This yields
\begin{align}
\begin{aligned} \label{D-5-0-1}
&\int_{\bbt^{2d}} (\partial_t \Psi_j(x_1, x_2) )\overline{\Psi_j(x_1, x_2)}dx_1 dx_2 =-\int_{\bbt^{2d}} \mathrm{i} {\mathcal H} \Psi_j(x_1, x_2)\overline{\Psi_j(x_1, x_2)} dx_1dx_2 \\
& \hspace{0.2cm} +\frac{\kappa_0}{N}\sum_{k=1}^N \int_{\bbt^{4d}} \left(\Psi_k(x_1, x_2^*)\overline{\Psi_j(x_1^*, x_2^*)}\Psi_j(x_1^*, x_2)-\Psi_j(x_1, x_2^*)\overline{\Psi_k(x_1^*, x_2^*)}\Psi_j(x_1^*, x_2)\right) \\
& \hspace{2.5cm} \times \overline{\Psi_j(x_1, x_2)}dx_1^*dx_2^*dx_1dx_2 \\
& \hspace{0.2cm}  +\frac{\kappa_1}{N}\sum_{k=1}^N\int_{\bbt^{4d}} \left(\Psi_j(x_1, x_2^*)\overline{\Psi_j(x_1^*, x_2^*)}\Psi_k(x_1^*, x_2)-\Psi_j(x_1, x_2^*)\overline{\Psi_k(x_1^*, x_2^*)}\Psi_j(x_1^*, x_2)\right) \\
& \hspace{2.5cm} \times \overline{\Psi_j(x_1, x_2)}dx_1^*dx_2^*dx_1dx_2.
\end{aligned}
\end{align}
We use ${\mathcal H}^{\dagger} = {\mathcal H}$, \eqref{D-5-0-0} and \eqref{D-5-0-1} to get 
\[
\int_{\bbt^{2d}} (\partial_t \Psi_j(x_1, x_2) )\overline{\Psi_j(x_1, x_2)}dx_1dx_2 +(c.c.)=0.
\]
This yields the desired estimate.
\end{proof}

\vspace{0.5cm}

Consider the Cauchy problem to the nonlinear system associated with \eqref{D-5}:
\begin{equation}\label{D-5-0}
\begin{cases}
\displaystyle \mathrm{i}\partial_t \varphi_j(x_1, x_2) \\
\displaystyle  \hspace{0.5cm} = \mathrm{i}\kappa_0 \int_{\bbt^d} \Big( \Big \langle \varphi_j(x_2^*) \Big | \varphi_j(x_2)  \Big \rangle \varphi_c(x_1, x_2^*)   
 - \Big \langle \varphi_c(x_2^*)  \Big | \varphi_j(x_2) \Big \rangle \varphi_j(x_1, x_2^*)  \Big) dx_2^*\\
\displaystyle  \hspace{0.7cm} + \mathrm{i}\kappa_1 \int_{\bbt^d} \Big(  \Big \langle \varphi_i(x_2^*) \Big | \varphi_c(x_2) \Big \rangle - \Big \langle \varphi_c(x_2^*) \Big| \varphi_j(x_2) \Big \rangle \Big) \varphi_j(x_1, x_2^*) dx_2^*, \\
\displaystyle \varphi_j \Big|_{t = 0} = \Psi_j^{in}.
\end{cases}
\end{equation}
Now we will show the solution splitting property of the SLM model \eqref{D-5}.
\begin{proposition}\label{P4.2}
Suppose that the one-body potential $V$ is additive in the sense that
\begin{equation} \label{D-5-1}
 V(x_1, x_2) = V_1(x_1) + V_2(x_2),
 \end{equation}
and let $\{\Psi_j \}$ and $\{\varphi_j \}$ be global smooth solutions to \eqref{D-5} and \eqref{D-5-0}, respectively. Then, one has 
\[  \Psi_j(t, x_1, x_2) =  e^{-\mathrm{i} {\mathcal H} t} \varphi_j(t,x_1, x_2), \quad j = 1, \cdots, N.   \]
\end{proposition}
\begin{proof}
First, we claim:
\begin{equation}\label{D-5-1-1}
e^{-\mathrm{i} {\mathcal H} (x_1, x_2)t}=e^{-\mathrm{i}{\mathcal H}_1(x_1)t}e^{-\mathrm{i} {\mathcal H}_2(x_2)t} =  e^{-\mathrm{i}{\mathcal H}_2(x_2)t} e^{-\mathrm{i}{\mathcal H}_1(x_1)t},
\end{equation}
where ${\mathcal H}_i = -\frac{1}{2} \Delta_{x_i} + V_i(x_i)$. \newline

\noindent {\it Proof of \eqref{D-5-1-1}}:  For any $C^2$-test function $f = f(x_1, x_2)$, we use \eqref{D-5-1} and definition of ${\mathcal H}$ to get 
\begin{align}
\begin{aligned} \label{D-5-2-1}
{\mathcal H} f(x_1, x_2) &= \Big(-\frac{1}{2} \Delta_{x_1} f(x_1, x_2) +V_1(x_1) f(x_1, x_2)  \Big) \\
&+ \Big (- \frac{1}{2} \Delta_{x_2} f(x_1, x_2) +V_2(x_2) f(x_1, x_2) \Big) \\
&={\mathcal H}_1 f(x_1, x_2)+ {\mathcal H}_2 f(x_1, x_2)  \\
&= {\mathcal H}_2 f(x_1, x_2)+ {\mathcal H}_1f(x_1, x_2)
\end{aligned}
\end{align}
and 
\begin{align}
\begin{aligned} \label{D-5-2-2}
&({\mathcal H}_1 \circ {\mathcal H}_2) f(x_1, x_2) = {\mathcal H}_1\Big (-\frac{1}{2} \Delta_{x_2} f(x_1, x_2) + V_2(x_2) f(x_1, x_2) \Big) \\
&=\frac{1}{4}\Delta_{x_1}\Delta_{x_2}f(x_1, x_2)-\frac{1}{2}V_1(x_1)\Delta_{x_2}f(x_1, x_2)-\frac{1}{2}V_2(x_2)\Delta_{x_1}f(x_1, x_2)+V_1(x_1)V_2(x_2)f(x_1, x_2)\\
&=\frac{1}{4}\Delta_{x_2}\Delta_{x_1}f(x_1, x_2)-\frac{1}{2}V_2(x_2)\Delta_{x_1}f(x_1, x_2)-\frac{1}{2}V_1(x_1)\Delta_{x_2}f(x_1, x_2)+V_2(x_2)V_1(x_1)f(x_1, x_2)\\
& = ({\mathcal H}_2 \circ {\mathcal H}_1) f(x_1, x_2).
\end{aligned}
\end{align}
Now, we use \eqref{D-5-2-1} and \eqref{D-5-2-2} to get the desired estimate \eqref{D-5-1-1}:
\begin{align*}\label{D-5-3}
e^{-\mathrm{i} {\mathcal H}t} = e^{-\mathrm{i}{\mathcal H}_1 t -{\mathrm i} {\mathcal H}_2 t} =e^{-\mathrm{i} {\mathcal H}_1t}e^{-\mathrm{i} {\mathcal H}_2t} = e^{-\mathrm{i} {\mathcal H}_2t} e^{-\mathrm{i} {\mathcal H}_1t}.
\end{align*}

\vspace{0.5cm}

\noindent $\bullet$~Step B: For $t \in \bbr_+,~~x_i \in \bbt^d, \quad j = 1, \cdots, N$, we set
\begin{equation*} \label{D-5-4}
\varphi_j(t,x_1, x_2) := e^{\mathrm{i} {\mathcal H} t}\Psi_j(t,x_1, x_2), \quad \mbox{or} \quad \Psi_j(t,x_1, x_2) = e^{-\mathrm{i} {\mathcal H} t} \varphi_j(t,x_1, x_2).
\end{equation*} 
Suppose that $\Psi_j$ satisfies system \eqref{D-5}. Then, it suffices to show that $\varphi_j$ satisfies \eqref{D-5-0}. For this, we multiply $e^{\mathrm{i} {\mathcal H} t}$ to \eqref{D-5} and compare the L.H.S. and R.H.S of the resulting relation to derive \eqref{D-5-0} for $\varphi_j$.  \newline

\noindent $\diamond$~(Estimate of L.H.S.): By direct calculation, one has 
\begin{equation} \label{D-5-4-1}
e^{\mathrm{i} {\mathcal H} t}\ \mathrm{i}\partial_t \Psi_j - e^{\mathrm{i} {\mathcal H} t} {\mathcal H} \Psi_j = {\mathrm i}  \partial_t (e^{\mathrm{i} {\mathcal H} t} \Phi_j ) = {\mathrm i}  \partial_t \varphi_j.
\end{equation}

\vspace{0.5cm}

\noindent $\diamond$~(Estimate of R.H.S.): Recall the R.H.S.:
\begin{align}
\begin{aligned}\label{D-5-5}
&{\mathrm{i}\kappa_0}e^{\mathrm{i}H t} \int_{\bbt^{d}} \Big(  \Big \langle \Psi_j(x_2^*) \Big | \Psi_j(x_2) \Big \rangle  \Psi_c(x_1, x_2^*) -  \Big \langle \Psi_c(x_2^*) \Big | \Psi_j(x_2) \Big \rangle \Psi_j(x_1, x_2^*) \Big) dx_2^* \\
&+{\mathrm{i}\kappa_1}e^{\mathrm{i}Ht} \int_{\bbt^{d}} \Big( \Big \langle \Psi_j(x_2^*) \Big| \Psi_c(x_2) \Big \rangle 
-  \Big \langle \Psi_c(x_2^*) \Big| \Psi_j(x_2) \Big \rangle \Big) \Psi_j(x_1, x_2^*) dx_2^*.
\end{aligned}
\end{align}
Then, we use the isometry of $e^{\mathrm{i} {\mathcal H} t}$ to rewrite \eqref{D-5-5} as 
\begin{align}
\begin{aligned}\label{D-5-6}
&{\mathrm{i}\kappa_0} \int_{\bbt^{d}} \Big(  \Big \langle \varphi_j(x_2^*) \Big | \varphi_j(x_2) \Big \rangle  \varphi_c(x_1, x_2^*) -  \Big \langle \varphi_c(x_2^*) \Big | \varphi_j(x_2) \Big \rangle \varphi_j(x_1, x_2^*) \Big) dx_2^* \\
&+{\mathrm{i}\kappa_1} \int_{\bbt^{d}} \Big( \Big \langle \varphi_j(x_2^*) \Big| \varphi_c(x_2) \Big \rangle \varphi_j(x_1, x_2^*) 
-  \Big \langle \varphi_c(x_2^*) \Big| \varphi_j(x_2) \Big \rangle \varphi_j(x_1, x_2^*) \Big )dx_2^*.
\end{aligned}
\end{align}
Finally, we combine \eqref{D-5-4-1} and \eqref{D-5-6} to get the desired system \eqref{D-5-0}. 
\end{proof}

\subsection{Reduction to the LM model} \label{sec:4.3} Next, we discuss how system \eqref{D-5} can be reduced to the generalized Lohe matrix model \eqref{B-4} using the same strategy in Section \ref{sec:4.1}  \newline 

Suppose that ${\mathcal H}$ is additive without any interaction Hamiltonians:
\[ {\mathcal H} = {\mathcal H}_1 + {\mathcal H}_2, \quad {\mathcal H}_i = -\frac{1}{2} \Delta_{x_i}  + V_i(x_i), \quad x_i \in \bbt^d, \quad  i = 1, 2. \]

From the chapter 11 of the book \cite{Har} we know that the eigenfunctions of the Hermitian operator $\mathcal{H}$ forms the orthonormal basis. Let $\{\phi^1_{\alpha}(x_1)\}_{\alpha_1 = 1}^{\infty}$ and $\{\phi^2_{\alpha_2} (x_2)\}_{\alpha_2 = 1}^{\infty}$ be two orthonormal basis consisting of eigenfunctions of ${\mathcal H}_i$:
\[ {\mathcal H}_1 \phi^1_{\alpha_1} =E^1_{\alpha_1} \phi^1_{\alpha_1} \quad \mbox{and} \quad  {\mathcal H}_2 \phi^2_{\alpha_2}=E^2_{\alpha_2} \phi^2_{\alpha_2}. \]
Now, we introduce standing wave solutions $\Phi^1_{\alpha_1}$ and $\Phi^2_{\alpha_2}$ as follows:
\[
\Phi^1_{\alpha_1} (t, x_1) := e^{-\mathrm{i}E^1_{\alpha}t} \phi^1_{\alpha_1} (x_1) \quad \mbox{and} \quad \Phi^2_{\alpha_2}(t, x_2) := e^{-\mathrm{i}E^2_{\alpha}t}  \phi^2_{\alpha_2}(x_2).
\]
Then, it is easy to see
\begin{align}
\begin{aligned}  \label{D-5-7}
& \mathrm{i}\partial_t \Phi^1_{\alpha_1}=E_\alpha^1 \Phi^1_{\alpha_1},\quad \mathrm{i}\partial_t \Phi^2_{\alpha_2} =E_{\alpha_2}^2 \Phi^2_{\alpha_2}, \\
& {\mathcal H}_1 \Phi^1_{\alpha_1} =E^1_{\alpha_1} \Phi^1_{\alpha_1}, \quad {\mathcal H}_2 \Phi^2_{\alpha_2}=E^2_{\alpha_2} \Phi^2_{\alpha_2}.
\end{aligned}
\end{align}
Due to \eqref{D-5-7}, the tensor product 
\[ (\Phi^1_{\alpha_1} \otimes \Phi^2_{\alpha_2})(t, x_1, x_2) := \Phi^1_{\alpha_1}(t, x_1)  \Phi^2_{\alpha_2}(t, x_2) \]
satisfies  two-dimensional linear Schr\"{o}dinger equation with the Hamiltonian ${\mathcal H}$:
\begin{align}\label{D-6}
\mathrm{i}\partial_t ( \Phi^1_{\alpha_1} \otimes \Phi^2_{\alpha_2})=(E_{\alpha_1}^1+E_{\alpha_2}^2)(\Phi^1_{\alpha_1} \otimes \Phi^2_{\alpha_2}) = {\mathcal H} ( \Phi^1_{\alpha_1} \otimes \Phi^2_{\alpha_2}).
\end{align}
Now, we expand $\Psi_j = \Psi_j(t, x_1, x_2)$ in terms of the basis $\{ \Phi^1_{\alpha_1} \otimes \Phi^2_{\alpha_2} \}_{\alpha_1, \alpha_2}$:
\begin{equation} \label{D-7}
\Psi_j =\sum_{\alpha_1,\alpha_2} [A_j(t)]_{\alpha_1 \alpha_2}  \left(\Phi^1_{\alpha_1} \otimes \Phi^2_{\alpha_2}\right).
\end{equation}
Here we set $A_j = ([A_j]_{\alpha \beta})$ to be an infinite matrix (see a review paper \cite{S-S} for theory of infinite matrices). 

\begin{proposition} \label{P4.3}
Let $\{\Psi_j \}$ be a global smooth solution to \eqref{D-5}, and let $A_j = ([A_j]_{\alpha_1 \alpha_2})$ be an infinite matrix whose elements is given as a coefficient in \eqref{D-7}. Then, the matrix ensemble $\{A_j \}$ satisfies the generalized Lohe matrix model on $(\ell^\infty \cap \ell^2)(\bbz^2_+)$:
\[
\dot{A}_j= \kappa_0(A_c A_j^* A_j- A_j A_c^* A_j)  + \kappa_1 (A_j A_j^* A_c -A_j A_c^* A_j), \quad j = 1, \cdots, N.
\]
\end{proposition}
\begin{proof}
We differentiate \eqref{D-7} with respect to $t$ and use \eqref{D-6} to derive
\begin{align}
\begin{aligned}\label{D-8}
\mathrm{i} \partial_t \Psi_j &=\sum_{\alpha_1,\alpha_2}\Big[ [{\dot A}_j]_{\alpha_1\alpha_2} \mathrm{i}(\Phi^1_{\alpha_1} \otimes \Phi^2_{\alpha_2}) + [A_j]_{\alpha_1\alpha_2} \mathrm{i}\partial_t(\Phi^1_{\alpha_1} \otimes \Phi^2_{\alpha_2}) \Big] \\
&=\sum_{\alpha_1,\alpha_2} \Big[ [\dot{A}_j]_{\alpha_1 \alpha_2}\mathrm{i}(\Phi^1_{\alpha_1} \otimes \Phi^2_{\alpha_2} )+  [A_j]_{\alpha_1\alpha_2}  {\mathcal H} (\Phi^1_{\alpha_1} \otimes \Phi^2_{\alpha_2}) \Big].
\end{aligned}
\end{align}
We substitute \eqref{D-7} into \eqref{D-5} to find
\begin{align}
\begin{aligned}\label{D-9}
& \mathrm{i} \sum_{\alpha_1,\alpha_2} [{\dot A}_j]_{\alpha_1 \alpha_2} ( \Phi^1_{\alpha_1} \otimes \Phi^2_{\alpha_2} )(x_1, x_2)\\
& \hspace{0.2cm} = \mathrm{i}\kappa_0 \underbrace{\int_{\bbt^d} \Big( \Big \langle \Psi_j(x_2^*) \Big| \Psi_j(x_2) \Big \rangle \Psi_c(x_1, x_2^*) 
-  \Big \langle \Psi_c(x_2^*) \Big| \Psi_j(x_2) \Big \rangle \Psi_j(x_1, x_2^*) \Big) dx_2^*}_{=:\mathcal{I}_{11}} \\
& \hspace{0.2cm} + \mathrm{i}\kappa_1 \underbrace{\int_{\bbt^d} \Big(  \Big \langle \Psi_j(x_2^*) \Big| \Psi_c(x_2) \Big \rangle \Psi_j(x_1, x_2^*) -  \Big \langle \Psi_c(x_2^*) \Big| \Psi_j(x_2) \Big \rangle \Psi_j(x_1, x_2^*) \Big ) dx_2^*}_{=:\mathcal{I}_{12}}.
\end{aligned}
\end{align}
Below, we estimate ${\mathcal I}_{1i}$ separately. \newline

\noindent $\bullet$~(Estimate of ${\mathcal I}_{11}$): By direct estimate, one has
\begin{align}
\begin{aligned} \label{D-10}
\mathcal{I}_{11} &= \int_{\bbt^d} \left( \Big \langle \Psi_j(x_2^*) \Big| \Psi_j(x_2) \Big \rangle \Psi_c(x_1, x_2^*) 
-  \Big \langle \Psi_c(x_2^*) \Big| \Psi_j(x_2) \Big \rangle \Psi_j(x_1, x_2^*) \right) dx_2^* \\
&=\sum_{\alpha, \beta, \gamma, \delta, \epsilon,\eta} \int_{\bbt^{2d}} \left([A_c]_{\alpha\beta}[\bar{A}_j]_{\gamma\delta} [A_j]_{\epsilon\eta}-[A_j]_{\alpha\beta}[\bar{A}_c]_{\gamma\delta} [A_j]_{\epsilon\eta}\right) \\
& \hspace{3cm} \times \Phi^1_\alpha(x_1)\Phi^2_\beta(x_2^*)\overline{\Phi^1_\gamma(x_1^*)\Phi^2_\delta(x_2^*)} \Phi^1_\epsilon(x_1^*)\Phi^2_\eta(x_2) dx_1^*dx_2^*\\
&=\sum_{\alpha, \beta, \gamma, \delta, \epsilon,\eta} \left( [A_{c}]_{\alpha\beta} [\bar{A}_{j}]_{\gamma\delta} [A_{j}]_{\epsilon\eta}- [A_j]_{\alpha\beta}[\bar{A}_{c}]_{\gamma\delta} [A_j]_{\epsilon\eta} \right) \Phi^1_\alpha(x_1) \Phi^2_\eta(x_2) \delta_{\gamma\epsilon}\delta_{\beta\delta}\\
&=\sum_{\alpha,\beta,\gamma,\eta}( [A_c]_{\alpha\beta}[\bar{A}_j]_{\gamma\beta} [A_j]_{\gamma\eta}- [A_j]_{\alpha\beta}
[\bar{A}_c]_{\gamma\beta} [A_j]_{\gamma\eta})  ( \Phi^1_{\alpha} \otimes \Phi^2_{\eta} )(x_1, x_2).
\end{aligned}
\end{align}

\vspace{0.2cm}

\noindent $\bullet$~(Estimate of ${\mathcal I}_{12}$): Similarly, one has
\begin{equation} \label{D-11}
\mathcal{I}_{12} =\sum_{\alpha,\eta}[A_j A_j^* A_c - A_j A_c^* A_j]_{\alpha\eta} (\Phi^1_\alpha \otimes \Phi^2_\eta)(x_1, x_2).
\end{equation}
We combine \eqref{D-9}, \eqref{D-10} and \eqref{D-11} to get 
\[  \sum_{\alpha_1,\alpha_2} [\dot{A}_{j}(t)]_{\alpha_1\alpha_2} \mathrm{i}  (\Phi^1_{\alpha_1} \otimes \Phi^2_{\alpha_2}) = \mathrm{i} \sum_{\alpha,\eta}\Big(\kappa_0[ A_c A_j^* A_j - A_j A_c^* A_j]_{\alpha\eta}+\kappa_1[A_j A_j^* A_c-A_j A_c^* A_j]_{\alpha\eta} \Big) (\Phi^1_{\alpha} \otimes \Phi^2_{\eta}).
\]
Then, we use the orthonormality of $\{ \Phi^1_{\alpha_1} \otimes \Phi^2_{\alpha_2} \}$ to see 
\[
[\dot{A}_{j}]_{\alpha\beta} =  \kappa_0 [ A_c A_j^* A_j -A_jA_c^* A_j]_{\alpha\beta}+\kappa_1[A_j A_j^* A_c-A_j A_c^* A_j]_{\alpha\beta}.
\]
This yields
\[
\dot{A}_j= \kappa_0( A_c A_j^* A_j - A_j A_c^* A_j)+\kappa_1(A_j A_j^* A_c-A_j A_c^* A_j).
\]
\end{proof}
\subsection{Emergent dynamics}  \label{sec:4.4}
In this subsection, we introduce an order parameter and study emergent dynamics of \eqref{D-5}.  For a given ensemble $\{\Psi_i  = \Psi_i(x_1, x_2) \}$, we set 
\[ \Psi_c :=\frac{1}{N}\sum_{k=1}^N\Psi_k \quad \mbox{and} \quad R := \|\Psi_c \|_{2}.   \]
\begin{lemma}\label{L4.1}
Let $\{\Psi_i\}$ be a global smooth solution to \eqref{D-5}. Then, the order parameter $R$ satisfies
\begin{align*}
\begin{aligned}
& (i)~\frac{dR^2}{dt} =\frac{\kappa_0}{N}\sum_{i=1}^N\int_{\bbt^{2d}} \left|\int_{\bbt^d}(\overline{\Psi_c(x_1, x_2)}\Psi_i(x_1^*, x_2)-\overline{\Psi_i(x_1, x_2)}\Psi_c(x_1^*, x_2))dx_2\right|^2dx_1dx_1^*\\
& \hspace{1cm} +\frac{\kappa_1}{N}\sum_{i=1}^N\int_{\bbt^{2d}} \left|\int_{\bbt^d}(\overline{\Psi_c(x_1, x_2)}\Psi_i(x_1, x_2^*)-\overline{\Psi_i(x_1, x_2)}\Psi_c(x_1, x_2^*))dx_1\right|^2dx_2dx_2^*\geq0. \\
& (ii)~\frac{d}{dt} \sum_{i, j=1}^N\|\Psi_i-\Psi_j\|^2_{2}\\
&\hspace{0.5cm}=-2\kappa_0N\sum_{i=1}^N\int_{\bbt^{2d}} \left|\int_{\bbt^d}(\overline{\Psi_c(x_1, x_2)}\Psi_i(x_1^*, x_2)-\overline{\Psi_i(x_1, x_2)}\Psi_c(x_1^*, x_2))dx_2\right|^2dx_1dx_1^*\\
&\hspace{1cm} -2\kappa_1N\sum_{i=1}^N\int_{\bbt^{2d}} \left|\int_{\bbt^d}(\overline{\Psi_c(x_1, x_2)}\Psi_i(x_1, x_2^*)-\overline{\Psi_i(x_1, x_2)}\Psi_c(x_1, x_2^*))dx_1\right|^2dx_2dx_2^*\leq0.
\end{aligned}
\end{align*}
\end{lemma}
\begin{proof}
\noindent (i)~Note that 
\[
\frac{d}{dt}\langle\Psi_i|\Psi_j\rangle=\langle\partial_t\Psi_i|\Psi_j\rangle+\langle\Psi_i|\partial_t\Psi_j\rangle.
\]
Below, we estimate the second term $\langle\Psi_i|\partial_t\Psi_j\rangle.$  By direct calculation, one has
\begin{align*}
\begin{aligned}
\langle\Psi_i|\partial_t\Psi_j\rangle &=\int_{\bbt^{2d}}  \overline{\Psi_i(x_1, x_2)}\partial_t\Psi_j(x_1, x_2)dx_1dx_2  =\underbrace{\int_{\bbt^{2d}} (-\mathrm{i})\overline{ \Psi_i(x_1, x_2)}(\mathcal{H}\Psi_j)(x_1, x_2)dx_1dx_2}_{:=\mathcal{I}_{21}^{ij}}\\
& + \kappa_0 \underbrace{\int_{\bbt^{4d}} \overline{\Psi_i(x_1, x_2)}\Psi_c(x_1, x_2^*)\overline{\Psi_j(x_1^*, x_2^*)}\Psi_j(x_1^*, x_2)dx_1dx_2dx_1^*dx_2^*}_{:=\mathcal{I}_{22}^{ij}}\\
&-\kappa_0 \underbrace{\int_{\bbt^{4d}} \overline{\Psi_i(x_1, x_2)}\Psi_j(x_1, x_2^*)\overline{\Psi_c(x_1^*, x_2^*)}\Psi_j(x_1^*, x_2)dx_1dx_2dx_1^*dx_2^*}_{:=\mathcal{I}_{23}^{ij}}\\
&+ \kappa_1 \underbrace{\int_{\bbt^{4d}} \overline{\Psi_i(x_1, x_2)}\Psi_j(x_1, x_2^*)\overline{\Psi_j(x_1^*, x_2^*)}\Psi_c(x_1^*, x_2)dx_1dx_2dx_1^*dx_2^*}_{:=\mathcal{I}_{24}^{ij}}\\
&-\kappa_1 \underbrace{\int_{\bbt^{4d}} \overline{\Psi_i(x_1, x_2)}\Psi_j(x_1, x_2^*)\overline{\Psi_c(x_1^*, x_2^*)}\Psi_j(x_1^*, x_2 )dx_1dx_2dx_1^*dx_2^*}_{:=\mathcal{I}_{25}^{ij}}.
\end{aligned}
\end{align*}
The term $\mathcal{I}_{21}^{ij}$ will be cancelled with a similar term in $\langle\partial_t\Psi_i|\Psi_j\rangle$ due to the Hermitian property of ${\mathcal H}$. Then we have
\begin{align*}
\frac{d}{dt} \sum_{i, j=1}^N\langle\Psi_i|\Psi_j\rangle =\kappa_0\sum_{i, j=1}^N\left(\mathcal{I}_{22}^{ij} +\bar{\mathcal{I}}_{22}^{ij}-\mathcal{I}_{23}^{ij}-\bar{\mathcal{I}}_{23}^{ij}\right)+\kappa_1\sum_{i, j=1}^N\left(\mathcal{I}_{24}^{ij} +\bar{\mathcal{I}}_{24}^{ij}-\mathcal{I}_{25}^{ij}-\bar{\mathcal{I}}_{25}^{ij}\right).
\end{align*}
The other terms ${\mathcal I}_{2k}^{ij},~k=2, \cdots, 5$ can be treated as follows.
\begin{align*}
\begin{aligned} 
\sum_{i, j=1}^N\mathcal{I}_{22}^{ij} &=N\sum_{i=1}^N\int_{\bbt^{4d}} \overline{\Psi_c(x_1, x_2)}\Psi_c(x_1, x_2^*)\overline{\Psi_i(x_1^*, x_2^*)}\Psi_i(x_1^*, x_2)dx_1dx_2dx_1^*dx_2^*, \\
\sum_{i, j=1}^N\mathcal{I}_{23}^{ij} &=N\sum_{i=1}^N\int_{\bbt^{4d}} \overline{\Psi_c(x_1, x_2)}\Psi_i(x_1, x_2^*)\overline{\Psi_c(x_1^*, x_2^*)}\Psi_i(x_1^*, x_2)dx_1dx_2dx_1^*dx_2^*, \\
\sum_{i, j=1}^N\mathcal{I}_{24}^{ij} &=N\sum_{i=1}^N\int_{\bbt^{4d}} \overline{\Psi_c(x_1, x_2)}\Psi_i(x_1, x_2^*)\overline{\Psi_i(x_1^*, x_2^*)}\Psi_c(x_1^*, x_2)dx_1dx_2 dx_1^*dx_2^*, \\
\sum_{i, j=1}^N\mathcal{I}_{25}^{ij} &=N\sum_{i=1}^N\int_{\bbt^{4d}}\overline{\Psi_c(x_1, x_2)}\Psi_i(x_1, x_2^*)\overline{\Psi_c(x_1^*, x_2^*)}\Psi_i(x_1^*, x_2)dx_1dx_2dx_1^*dx_2^*.
\end{aligned}
\end{align*}
Note that 
\begin{align*}
\begin{aligned}
& \sum_{i, j=1}^N\left(\mathcal{I}_{22}^{ij} +\bar{\mathcal{I}}_{22}^{ij}-\mathcal{I}_{23}^{ij}-\bar{\mathcal{I}}_{23}^{ij}\right) \\
& \hspace{1cm} =-N\sum_{i=1}^N\int_{\bbt^{4d}} (\overline{\Psi_c(x_1, x_2)}\Psi_i(x_1^*, x_2)-\overline{\Psi_i(x_1, x_2)}\Psi_c(x_1^*, x_2)) \\
& \hspace{1.2cm} \times (\overline{\Psi_c(x_1^*, x_2^*)}\Psi_i(x_1, x_2^*)-\overline{\Psi_i(x_1^*, x_2^*)}\Psi_c(x_1, x_2^*))dx_1 dx_2 dx_1^*dx_2^*\\
& \hspace{1cm} =-N\sum_{i=1}^N\int_{\bbt^{2d}}\left(\int_{\bbt^d}(\overline{\Psi_c(x_1, x_2)}\Psi_i(x_1^*, x_2)-\overline{\Psi_i(x_1, x_2)}\Psi_c(x_1^*, x_2))dx_2 \right)\\
&\hspace{1.2cm}\times\left(\int_{\bbt^2}(\overline{\Psi_c(x_1^*, x_2^*)}\Psi_i(x_1, x_2^*)-\overline{\Psi_i(x_1^*, x_2^*)}\Psi_c(x_1, x_2^*))dx_2^*\right)dx_1dx_1^*\\
& \hspace{1cm} =N\sum_{i=1}^N\int_{\bbt^{2d}} \left|\int_{\bbt^d}(\overline{\Psi_c(x_1, x_2)}\Psi_i(x_1^*, x_2)-\overline{\Psi_i(x_1, x_2)}\Psi_c(x_1^*, x_2))dx_2\right|^2dx_1dx_1^*\geq0.
\end{aligned}
\end{align*}
Similarly, we have
\begin{align*}
\begin{aligned}
& \sum_{i, j=1}^N\left(\mathcal{I}_{24}^{ij}+\bar{\mathcal{I}}_{24}^{ij}-\mathcal{I}_{25}^{ij}-\bar{\mathcal{I}}_{25}^{ij}\right) \\
& \hspace{0.5cm} =N\sum_{i=1}^N\int_{\bbt^{2d}} \left|\int_{\bbt^d}(\overline{\Psi_c(x_1, x_2)}\Psi_i(x_1, x_2^*)-\overline{\Psi_i(x_1, x_2)}\Psi_c(x_1, x_2^*))dx_1\right|^2dx_2dx_2^*\geq0.
\end{aligned}
\end{align*}
Finally we have
\begin{align*}
&\frac{d}{dt}\left(\sum_{i, j=1}^N\langle\Psi_i|\Psi_j\rangle\right) =\kappa_0N\sum_{i=1}^N\int_{\bbt^{2d}} \left|\int_{\bbt^d}(\overline{\Psi_c(x_1, x_2)}\Psi_i(x_1^*, x_2)-\overline{\Psi_i(x_1, x_2)}\Psi_c(x_1^*, x_2))dx_2\right|^2dx_1dx_1^*\\
& \hspace{0.5cm}+\kappa_1N\sum_{i=1}^N\int_{\bbt^{2d}} \left|\int_{\bbt^d}(\overline{\Psi_c(x_1, x_2)}\Psi_i(x_1, x_2^*)-\overline{\Psi_i(x_1, x_2)}\Psi_c(x_1, x_2^*))dx_1\right|^2dx_2dx_2^*
\geq0.
\end{align*}
(ii)~By direct calculation, one has
\begin{align*}
\begin{aligned}
&\frac{d}{dt}\left(\sum_{i, j=1}^N\|\Psi_i-\Psi_j\|^2_2\right) =\frac{d}{dt}\left(\sum_{i, j=1}^N(2-\langle{\Psi_i|\Psi_j}\rangle-\langle{\Psi_j|\Psi_i}\rangle )\right)=-2\frac{d}{dt}\left(\sum_{i, j=1}^N\langle{\Psi_i|\Psi_j}\rangle\right)\\
& \hspace{0.2cm} =-2\kappa_0N\sum_{i=1}^N\int_{\bbt^{2d}} \left|\int_{\bbt^d}(\overline{\Psi_c(x_1, x_2)}\Psi_i(x_1^*, x_2)-\overline{\Psi_i(x_1, x_2)}\Psi_c(x_1^*, x_2))dx_2\right|^2dx_1dx_1^*\\
& \hspace{0.2cm} -2\kappa_1N\sum_{i=1}^N\int_{\bbt^{2d}} \left|\int_{\bbt^d}(\overline{\Psi_c(x_1, x_2)}\Psi_i(x_1, x_2^*)-\overline{\Psi_i(x_1, x_2)}\Psi_c(x_1, x_2^*))dx_1\right|^2dx_2dx_2^*\leq0.
\end{aligned}
\end{align*}
\end{proof}
\begin{theorem} \label{T4.1}
Let $\{\Psi_i\}$ be a global smooth solution to \eqref{D-5} with the initial data satisfying $R^{in} >0$. Then we have
\begin{eqnarray*}
&& (i)~R(t) \geq R^{in}, \quad t > 0 \quad \mbox{and} \quad \lim_{t \to \infty} |{\dot R}(t)| = 0. \cr
&& (ii)~\lim_{t \to \infty} \int_{\bbt^{2d}}\Big|  \langle \Psi_c(x_1) | \Psi_j(x_1^*) \rangle -    \langle \Psi_j(x_1)   | \Psi_c(x^*_1)  \rangle \Big|^2 dx_1 dx_1^* = 0. \\
&& (iii)~\lim_{t \to \infty}  \int_{\bbt^{2d}}\Big|  \langle \Psi_c(x_2) | \Psi_j(x_2^*) \rangle -    \langle \Psi_j(x_2)   | \Psi_c(x^*_2)  \rangle \Big|^2 dx_2 dx_2^* = 0.
\end{eqnarray*}
\end{theorem}
\begin{proof} (i) It follows from Lemma \ref{L4.1} (i) that 
\begin{align*}
\begin{aligned}
\frac{dR^2}{dt} &= \frac{\kappa_0}{N} \sum_{i=1}^N\int_{\bbt^{2d}} \Big | \langle \Psi_c(x_1) | \Psi_j(x_1^*) \rangle -    \langle \Psi_j(x_1)   | \Psi_c(x^*_1)  \rangle \Big|^2 dx_1dx_1^* \\
&+ \frac{\kappa_1}{N} \sum_{i=1}^N\int_{\bbt^{2d}} \Big |  \langle \Psi_c(x_2) | \Psi_j(x_2^*) \rangle -    \langle \Psi_j(x_2)   | \Psi_c(x^*_2)  \rangle \Big |^2 dx_2 dx_2^* \geq 0.
\end{aligned}
\end{align*}
Thus, one has 
\[ R^2(t) \geq |R^{in}|^2, \quad \mbox{i.e.,} \quad  R(t) \geq R^{in}, \quad t > 0. \]
Since $R$ is non-decreasing and bounded by $1$, there exists $R^{\infty} \in [R^{in}, 1]$ such that 
\begin{equation*} \label{NN-1}
\lim_{t \to \infty} R(t) = R^{\infty}.
\end{equation*}
\end{proof}
Next, we show that $\frac{d^2}{dt^2}R^2$ is uniformly bounded in $t$ so that $\frac{d}{dt}R^2$ is uniformly continuous.
\begin{lemma}
Let $\{\Psi_i\}_{i=1}^N$ be a global smooth solution of \eqref{D-5-1}. Then the second derivative of $R(t)^2=\|\Psi_c(t)\|_2^2$ is uniformly bounded in time.
\end{lemma}
\begin{proof} Note that
\begin{align*}
\begin{aligned}
\frac{d}{dt}R^2 &= \frac{\kappa_0}{N} \sum_{j=1}^N\int_{\bbt^{2d}} \Big | \langle \Psi_c(x_1) | \Psi_j(x_1^*) \rangle -    \langle \Psi_j(x_1)   | \Psi_c(x^*_1)  \rangle \Big|^2 dx_1dx_1^* \\
&+ \frac{\kappa_1}{N} \sum_{j=1}^N\int_{\bbt^{2d}} \Big |  \langle \Psi_c(x_2) | \Psi_j(x_2^*) \rangle -    \langle \Psi_j(x_2)   | \Psi_c(x^*_2)  \rangle \Big |^2 dx_2 dx_2^*.
\end{aligned}
\end{align*}
Then one has
\begin{align*}
\begin{aligned}
\frac{d^2}{dt^2}R^2 &= \frac{\kappa_0}{N} \sum_{j=1}^N\int_{\bbt^{2d}}\partial_t \Big| \langle \Psi_c(x_1) | \Psi_j(x_1^*) \rangle -    \langle \Psi_j(x_1)   | \Psi_c(x^*_1)  \rangle \Big|^2 dx_1dx_1^* \\
&\hspace{0.2cm} + \frac{\kappa_1}{N} \sum_{j=1}^N\int_{\bbt^{2d}}\partial_t \Big |  \langle \Psi_c(x_2) | \Psi_j(x_2^*) \rangle -    \langle \Psi_j(x_2)   | \Psi_c(x^*_2)  \rangle \Big |^2 dx_2 dx_2^*.
\end{aligned}
\end{align*}
Now we show the R.H.S. of the above relation is uniformly bounded. For this, note that
\begin{align}
\begin{aligned}\label{D-12}
&\partial_t \Big| \langle \Psi_c(x_1) | \Psi_j(x_1^*) \rangle -    \langle \Psi_j(x_1)   | \Psi_c(x^*_1)  \rangle \Big|^2\\
& \hspace{0.2cm} =\partial_t\big( (\langle \Psi_c(x_1) | \Psi_j(x_1^*) \rangle -    \langle \Psi_j(x_1)   | \Psi_c(x^*_1)  \rangle)(\langle  \Psi_j(x_1^*)|\Psi_c(x_1) \rangle -    \langle  \Psi_c(x^*_1)|\Psi_j(x_1)    \rangle)\big)\\
&  \hspace{0.2cm} =\frac{1}{N^2}\sum_{k=1}^N\partial_t\Big( (\langle \Psi_k(x_1) | \Psi_j(x_1^*) \rangle -    \langle \Psi_j(x_1)   | \Psi_k(x^*_1)  \rangle)(\langle  \Psi_j(x_1^*)|\Psi_l(x_1) \rangle -    \langle  \Psi_l(x^*_1)|\Psi_j(x_1)    \rangle)\Big).
\end{aligned}
\end{align}
Below we will show that each term in \eqref{D-12} is bounded. Since $\partial_t \Psi_i=-\mathrm{i}\mathcal{H}\Psi_i+\mbox{(Coupling terms)}$, we will decompose $\partial_t\Psi_i$ with two part. By direct calculation, (Coupling terms) part can be bounded as follows:
\begin{align*}
&\|\partial_t \Psi_i+\mathrm{i}\mathcal{H}\Psi_i\|_2\\
&\hspace{0.5cm} \leq\kappa_0\left\| \int_{\bbt^{2d}} \Big(  \Psi_c(x_1, x_2^*)     \overline{\Psi_j(x_1^*, x_2^*)}\Psi_j(x_1^*, x_2)  - \Psi_j(x_1, x_2^*)      \overline{\Psi_c(x_1^*, x_2^*)}\Psi_j(x_1^*, x_2)  \Big)dx_1^* dx_2^*\right\|_2\\
&\hspace{0.7cm}+\kappa_1\left\| \int_{\bbt^{2d}} \Big( \Psi_j(x_1, x_2^*)      \overline{\Psi_j(x_1^*, x_2^*)}\Psi_c(x_1^*, x_2)-\Psi_j(x_1, x_2^*)      \overline{\Psi_c(x_1^*, x_2^*)}\Psi_j(x_1^*, x_2) \Big) dx_1^*dx_2^*\right\|_2.
\end{align*}
From the H\"{o}lder inequality, we have the boundedness of $\|\partial_t\Psi_i+\mathrm{i}\mathcal{H}\Psi_i\|_2$. Now we will prove the boundedness of the terms involved with $\mathcal{H}\Psi_i$. In general, we have following result:
\begin{align}
\begin{aligned}\label{D-13}
&\int_{\mathbb{T}^{4d}}(\mathcal{H}\Psi_i(x_1, x_2))\overline{\Psi_j(x_1^*, x_2)}\Psi_k(x_1^*, x_2^*)\overline{\Psi_l(x_1, x_2^*)}dx_1dx_2dx_1^*dx_2^*\\
& \hspace{0.5cm}+\int_{\mathbb{T}^{4d}}\Psi_i(x_1, x_2)\overline{\Psi_j(x_1^*, x_2)}\Psi_k(x_1^*, x_2^*)(\mathcal{H}\overline{\Psi_l(x_1, x_2^*)})dx_1dx_2dx_1^*dx_2^*
\end{aligned}
\end{align}
is bounded for all $i, j, k, l$. From
\begin{align*}
\mathcal{H}\Psi_i(x_1, x_2)=-\frac{1}{2}(\Delta_{x_1}+\Delta_{x_2})\Psi_i(x_1, x_2)+(V_1(x_1)+V_2(x_2))\Psi_i(x_1, x_2),\\
\mathcal{H}\Psi_l(x_1, x_2^*)=-\frac{1}{2}(\Delta_{x_1}+\Delta_{x_2^*})\Psi_l(x_1, x_2^*)+(V_1(x_1)+V_2(x_2^*))\Psi_l(x_1, x_2^*),
\end{align*}
we can easily check that the term involved with $V_1(x_1)$ and $V_2(x_2)$ are bounded, since continuous functions $V_1$ and $V_2$ are defined on the compact set $\mathbb{T}^d$. Now it is sufficient to prove that
\[
\int_{\mathbb{T}^{4d}}\Delta_{x_1}(\Psi_i(x_1, x_2)\overline{\Psi_j(x_1^*, x_2)}\Psi_k(x_1^*, x_2^*)\overline{\Psi_l(x_1, x_2^*)})dx_1dx_2dx_1^*dx_2^* 
\]
is bounded. However we have
\begin{align*}
&\int_{\mathbb{T}^{4d}}\Delta_{x_1}(\Psi_i(x_1, x_2)\overline{\Psi_j(x_1^*, x_2)}\Psi_k(x_1^*, x_2^*)\overline{\Psi_l(x_1, x_2^*)})dx_1dx_2dx_1^*dx_2^* \\
& \hspace{0.5cm} =\int_{\mathbb{T}^{4d}}\nabla_{x_1}\cdot\Big(\nabla_{x_1} \big(\Psi_i(x_1, x_2)\overline{\Psi_j(x_1^*, x_2)}\Psi_k(x_1^*, x_2^*)\overline{\Psi_l(x_1, x_2^*)}\big)\Big)dx_1dx_2dx_1^*dx_2^*=0 
\end{align*}
due to the divergence theorem. So we know that the value of \eqref{D-13} is uniformly bounded over time. If we use this property, then we can obtain the uniform boundedness of \eqref{D-12}. Now, we combine this results and \eqref{D-12} to get the uniform boundedness of 
\[
\int_{\mathbb{T}^{2d}}\partial_t \Big| \langle \Psi_c(x_1) | \Psi_j(x_1^*) \rangle -    \langle \Psi_j(x_1)   | \Psi_c(x^*_1)  \rangle \Big|^2dx_1dx_1^*.
\]
From this result, we have the uniform boundedness of $\frac{d^2}{dt^2}R^2$ which implies the uniform continuity of $\frac{dR^2}{dt}$.  Finally we apply Barbalat's lemma to derive
\[ \lim_{t \to \infty} \frac{dR^2}{dt} = 0. \]
This gives the desired estimate.
\end{proof}

\section{The Schr\"{o}dinger-Lohe tensor(SLT) model} \label{sec:5}
\setcounter{equation}{0}
In this section, we introduce the highest object located in the vertex of the infinite-dimensional {\it Schr\"{o}dinger-Lohe hierarchy}, namely ``{\it the Schrodinger-Lohe tensor model}". The construction is pretty much the same as the construction of the Schr\"{o}dinger-Lohe matrix model in spirit in previous section. First, we choose a standing wave solution for each linear Schr\"{o}dinger equation with Hamiltonian ${\mathcal H}_i$ and then, take a tensor product on $(\bbt^{d})^m = \bbt^d \times \cdots \times \bbt^d$, and then express our solution $\Psi(t, x_1, \cdots, x_m)$ as a linear combination of $\Phi_{\alpha_1}^1\otimes \Phi_{\alpha_2}^2 \otimes \cdots \otimes \Phi_{\alpha_m}^m$ with time-dependent coefficients. Finally, we design our Schr\"{o}dinger-Lohe tesnor model suitably so that it can be reduced into the Lohe tensor model as in the Schr\"{o}dinger-Lohe matrix model.  
\subsection{Construction of SLT model} \label{sec:5.1} 
Let $x = (x_1, \cdots, x_m) \in \bbt^{md}$ and ${\mathcal H}$ be an interaction free Hamiltonian:
\[
{\mathcal H} = - \frac{1}{2} \sum_{j=1}^{m} \Delta_{x_j} +V(x_1, x_2, \cdots, x_m).
\]
For each $j = 1, \cdots, m$, we set $\{ \Phi^j_{\alpha_j}(t, x_j)\}$ be an orthonormal family of standing wave solutions to the linear Schr\"{o}dinger equation associated with the Hamiltonian ${\mathcal H}_j$.  Then, the set $\{ \Phi^1_{\alpha_1} \otimes \Phi^2_{\alpha_2} \otimes \cdots \otimes \Phi^m_{\alpha_m} \}$ is an orthonormal family of basis, and for a wave function $\Psi_i = \Psi_i(t, x_1, \cdots, x_m)$, we set 
\begin{align}\label{E-1}
\Psi_i(t, x)=\sum_{\alpha_*}[T_i(t)]_{\alpha_*}  (\Phi^1_{\alpha_1} \otimes \Phi^2_{\alpha_2} \otimes \cdots \otimes \Phi^m_{\alpha_m})(t, x), \quad t \geq 0,~~x \in \bbt^{md},
\end{align}
where $\alpha = (\alpha_1, \cdots, \alpha_m)$ and 
\[ (\Phi^1_{\alpha_1} \otimes \Phi^2_{\alpha_2} \otimes \cdots \otimes \Phi^m_{\alpha_m})(t, x_1, \cdots, x_m) :=  \Phi^1_{\alpha_1}(t, x_1) \cdot \Phi_{\alpha_2}^2(t, x_2) \cdots \Phi^m_{\alpha_m}(t, x_m). \]
For notational simplicity, we also introduce several handy notation. 
\begin{align*}
\begin{aligned}
& \Psi(x_{10}, \cdots, x_{m0})=\Psi(x_{*0}), \quad \Psi(x_{11}, \cdots, x_{m1})=\Psi(x_{*1}), \\
& \Psi(x_{1i_1}, \cdots, x_{mi_m})=\Psi(x_{*i_*}), \quad \Psi(x_{1(1-i_1)}, \cdots, x_{m(1-i_m)})=\Psi(x_{*(1-i_*)}), \\
& dx_{10}\cdots dx_{m0}=dx_{*0}, \quad dx_{11}\cdots dx_{m1}=dx_{*1}, \\
& dx_{1i_1}\cdots dx_{mi_m}=dx_{*i_*}, \quad dx_{1(1-i_1)}\cdots dx_{m(1-i_m)}=dx_{*(1-i_*)}.
\end{aligned}
\end{align*}
Now, we consider the Cauchy problem to the SLT model:
\begin{equation}
\begin{cases} \label{E-2}
\displaystyle \mathrm{i}\partial_t\Psi_j(x_{*0}) - {\mathcal H} \Psi_j(x_{*0}) \\
\displaystyle  \hspace{0.5cm} = \mathrm{i}\sum_{i_*\in\{0, 1\}^m}\kappa_{i_*}\int_{\bbt^{md}} \Big(\Psi_c(x_{*i_*})\overline{\Psi_j(x_{*1})}\Psi_j(x_{*(1-i_*)})-\Psi_j(x_{*i_*})\overline{\Psi_c(x_{*1})}\Psi_j(x_{*(1-i_*)}) \Big)dx_{*1}, \\
\displaystyle \Psi_j \Big|_{t = 0} = \Psi_j^{in}.
\end{cases}
\end{equation}
Note that the Schr\"{o}dinger-Lohe model is a special case of Schr\"{o}dinger-Lohe tensor model with $m=1$, and Schr\"{o}dinger-Lohe matrix model is a special case of Schr\"{o}dinger Lohe tensor model with $m=2$. If we consider the case $m=0$ of the Schr\"{o}dinger-Lohe model, then $\Psi_i$ is no longer a function of space, i.e.,  $\Psi_i$ is a complex-valued function defined on time domain. Hence, we can easily derive the Kuramoto model from the Lohe tensor model.

\begin{proposition} \label{P5.1}
Let $\{\Psi_j \}$ be a global smooth solution to \eqref{E-2}.  Then $L^2$-norm of $\Psi_i$ is a conserved quantity.
\[
\frac{d}{dt}\|\Psi_i(t) \|^2_{2}=0, \quad t > 0, \quad  i = 1, \cdots, N.
\]
\end{proposition}
\begin{proof} By direct calculation, one has 
\begin{align*}
\begin{aligned}
& \frac{d}{dt}\langle{\Psi_i|\Psi_i}\rangle =\langle{\Psi_i|\partial_t\Psi_i}\rangle+(c.c.) =\left \langle \Psi_i\Big|-\mathrm{i}H \Psi_i(x_{*0})  \right. \\
& \left. +\sum_{i_*\in\{0, 1\}^m}\kappa_{i_*}\int_{\bbt^{md}} \Big(\Psi_c(x_{*i_*})\overline{\Psi_i(x_{*1})}\Psi_i(x_{*(1-i_*)})-\Psi_i(x_{*i_*})\overline{\Psi_c(x_{*1})}\Psi_i(x_{*(1-i_*)}) \Big) dx_{*1}\right\rangle +(c.c.)\\
&=\sum_{i_*\in\{0, 1\}^m}\kappa_{i_*}\left\langle{\Psi_i\Big|\int_{\bbt^{md}} (\Psi_c(x_{*i_*})\overline{\Psi_i(x_{*1})}\Psi_i(x_{*(1-i_*)})-\Psi_i(x_{*i_*})\overline{\Psi_c(x_{*1})}\Psi_i(x_{*(1-i_*)}))dx_{*1}}\right\rangle +(c.c.) \\
&=\sum_{i_*\in\{0, 1\}^m}\kappa_{i_*}\int_{\bbt^{2md}}  \underbrace{(\overline{\Psi_i(x_{*0})}\Psi_c(x_{*i_*})\overline{\Psi_i(x_{*1})}\Psi_i(x_{*(1-i_*)})-\overline{\Psi_i(x_{*0})}\Psi_i(x_{*i_*})\overline{\Psi_c(x_{*1})}\Psi_i(x_{*(1-i_*)}))dx_{*1}dx_{*0}}_{=:\mathcal{I}_{31}}\\
&+\sum_{i_*\in\{0, 1\}^m}\kappa_{i_*}\int_{\bbt^{2md}}\underbrace{({\Psi_i(x_{*0})}\overline{\Psi_c(x_{*i_*})}{\Psi_i(x_{*1})}\overline{\Psi_i(x_{*(1-i_*)})}-{\Psi_i(x_{*0})}\overline{\Psi_i(x_{*i_*})}{\Psi_c(x_{*1})}\overline{\Psi_i(x_{*(1-i_*)})})dx_{*1}dx_{*0}}_{=:\mathcal{I}_{32}}.
\end{aligned}
\end{align*}
Since $x_{*0}$ and $x_{*1}$ are dummy variables, we can exchange
\[
x_{*1}\leftrightarrow x_{*i_*},\quad x_{*0}\leftrightarrow x_{*(1-i_*)}.
\]
Hence, we can obtain 
\[ \mathcal{I}_{31}=-\mathcal{I}_{32} \]
to get the desired estimate $\frac{d}{dt}\langle \Psi_i|\Psi_i \rangle=0.$
\end{proof}

\vspace{0.5cm}

Consider the Cauchy problem to the nonlinear system associated with \eqref{E-2}:
\begin{equation}
\begin{cases} \label{E-3}
\displaystyle \mathrm{i}\partial_t\varphi_i(x_{*0})  \\
\displaystyle \hspace{0.2cm} = \mathrm{i}\sum_{i_*\in\{0, 1\}^m}\kappa_{i_*}\int_{\bbt^{md}} \Big (\varphi_c(x_{*i_*})\overline{\varphi_i(x_{*1})}\varphi_i(x_{*(1-i_*)})-\varphi_i(x_{*i_*})\overline{\varphi_c(x_{*1})}\varphi_i(x_{*(1-i_*)}) \Big)dx_{*1}, \\
\displaystyle \varphi_i \Big|_{t = 0} = \Psi_i^{in}.
\end{cases}
\end{equation}
Next, we derive the solution splitting property of \eqref{E-3}.
\begin{proposition}\label{P5.2}
Suppose that the one-body potential $V$ is additive: 
\[
V(x_1, \cdots, x_m) =  \sum_{j = 1}^{m} V_j(x_j),
\]
and let $\{\Psi_i\}$ and $\{\varphi_i \}$ be global smooth solutions to \eqref{E-2} and \eqref{E-3}, respectively. Then, one has 
\[  \Psi_i(t, x_1, x_2) =  e^{-\mathrm{i} {\mathcal H}t} \varphi_i(t,x_1, x_2).   \]
\end{proposition}
\begin{proof} As in the proof of Proposition \ref{P4.2}, we have
\begin{align*}
\begin{aligned}
& {\mathcal H} = - \frac{1}{2} \sum_{j=1}^{m} \Delta_{x_j} +V(x_1, x_2, \cdots, x_m) =  \sum_{j=1}^{m} \Big(  - \frac{1}{2} \Delta_{x_j} + V_j(x_j)   \Big) =  \sum_{j=1}^{m} {\mathcal H}_j, \\
& {\mathcal H}_i \circ {\mathcal H}_j = {\mathcal H}_j \circ {\mathcal H}_i, \qquad \forall i, j \in \{1, \cdots, N \}.
\end{aligned}
\end{align*}
These relations imply
\[
e^{-\mathrm{i} {\mathcal H}(x_1, x_2, \cdots, x_m)t}=e^{-\mathrm{i}{\mathcal H}_{a_1}(x_{a_1})t}\cdot e^{-\mathrm{i} {\mathcal H}_{a_2}(x_{a_2})t}\cdot\cdots\cdot e^{-\mathrm{i} {\mathcal H}_{a_m}(x_{a_m})t},
\]
where $\{a_1, a_2, \cdots, a_m\}=\{1, 2, \cdots, m\}$. Therefore, one has 
\[
e^{-\mathrm{i} {\mathcal H}(x_{*0})t}=e^{-\mathrm{i} {\mathcal H}(x_{*i_*})t}\cdot e^{\mathrm{i} {\mathcal H}(x_{*1})t}\cdot e^{-\mathrm{i} {\mathcal H}(x_{*(1-i_*)})t}.
\]
It follows from \eqref{E-2} that 
\begin{align*}
\begin{aligned}
&\mathrm{i}\partial_t\Psi_i(x_{*0}) ={\mathcal H} \Psi_i(x_{*0}) \\
& \hspace{0.2cm} +\mathrm{i}\sum_{i_*\in\{0, 1\}^m}\kappa_{i_*}\int_{\bbt^{md}} \Big(\Psi_c(x_{*i_*})\overline{\Psi_i(x_{*1})}\Psi_i(x_{*(1-i_*)})-\Psi_i(x_{*i_*})\overline{\Psi_c(x_{*1})}\Psi_i(x_{*(1-i_*)}) \Big)dx_{*1}.
\end{aligned}
\end{align*}
Since we have
\begin{align*}
&  \mathrm{i}\partial_t\varphi_i(x_{*0})  = \partial_t(e^{\mathrm{i}{\mathcal H}t}\Psi_i(x_{*0}))\\
& \hspace{0.5cm} =\sum_{i_*\in\{0, 1\}^m}\kappa_{i_*}e^{\mathrm{i}{\mathcal H}t}\int_{\bbt^{md}} (\Psi_c(x_{*i_*})\overline{\Psi_i(x_{*1})}\Psi_i(x_{*(1-i_*)})-\Psi_i(x_{*i_*})\overline{\Psi_c(x_{*1})}\Psi_i(x_{*(1-i_*)}))dx_{*1}\\
& \hspace{0.5cm} =\sum_{i_*\in\{0, 1\}^m}\kappa_{i_*}\int_{\bbt^{md}}(e^{\mathrm{i} {\mathcal H}(x_{*i_*})t}\Psi_c(x_{*i_*})\overline{e^{\mathrm{i} {\mathcal H}(x_{*1})t}\Psi_i(x_{*1})}e^{\mathrm{i} {\mathcal H}(x_{*(1-i_*)})t}\Psi_i(x_{*(1-i_*)})\\
&\hspace{2cm}-e^{\mathrm{i}{\mathcal H}(x_{*i_*})t}\Psi_i(x_{*i_*})\overline{e^{\mathrm{i} {\mathcal H}(x_{*1})t}\Psi_c(x_{*1})}e^{\mathrm{i} {\mathcal H}(x_{*(1-i_*)})t}\Psi_i(x_{*(1-i_*)}))dx_{*1}\\
& \hspace{0.5cm}=\sum_{i_*\in\{0, 1\}^m}\kappa_{i_*}\int_{\bbt^{md}}\big(\varphi_c(x_ {*i_*})\overline{\varphi_i(x_{*1})}\varphi_i(x_{*(1-i_*)})-\varphi_i(x_{*i_*})\overline{\varphi_c(x_{*1})}\varphi_i(x_{*(1-i_*)}))dx_{*1}.
\end{align*}
\end{proof}

\subsection{Reduction to the LT model} \label{sec:5.2}
In what follows, we present a reduction of the SLM to the Lohe tensor model in \cite{H-P2}. The basic idea is the same as in Section \ref{sec:4.1} and Section \ref{sec:4.3} for rank-1 and rank-2 tensors.   \newline

Let $\{\phi^i_{\alpha_i}(x_i)\}_{\alpha_i = 1}^{\infty}$ be an orthonormal system consisting of eigenfunctions of ${\mathcal H}_i$:
\[ {\mathcal H}_i \phi^i_{\alpha_i} =E^i_{\alpha_i} \phi^i_{\alpha_i}, \quad i = 1, \cdots, N. \]
Now, we introduce standing wave solution $\Phi^i_{\alpha_i}$ as follows:
\[
\Phi^i_{\alpha_i} (t, x_i) := e^{-\mathrm{i}E^i_{\alpha_i} t} \phi^i_{\alpha_i} (x_i), \quad i = 1, \cdots, N. \]
Then, it is easy to see check 
\begin{align*}\label{E-3-1}
\mathrm{i}\partial_t ( \Phi^1_{\alpha_1} \otimes \Phi^2_{\alpha_2} \otimes \cdots \otimes \Phi^m_{\alpha_m}) = {\mathcal H} ( \Phi^1_{\alpha_1} \otimes \Phi^2_{\alpha_2} \otimes \cdots \otimes \Phi^m_{\alpha_m}).
\end{align*}
Now, we expand $\Psi_j = \Psi_j(t, x_1, \cdots, x_m)$ in terms of the basis $\{ \Phi^1_{\alpha_1} \otimes \Phi^2_{\alpha_2}\otimes\cdots\otimes \Phi^m_{\alpha_m}\}_{\alpha_*}$:
\begin{equation} \label{E-4}
\Psi_j = \sum_{\alpha_*} [T_j(t)]_{\alpha_*}(\Phi_{\alpha_1}^1 \otimes \Phi_{\alpha_2}^2 \otimes \cdots \otimes \Phi_{\alpha_m}^m). 
\end{equation}

\begin{proposition} 
Let $\{\Psi_j \}$ be a global smooth solution to \eqref{E-2}.Then, the coefficient $[T_j]_{\alpha}$ satisfies the Lohe tensor model:
\[
\frac{d}{dt}[T_j]_{\alpha_{*0}} = \sum_{i_* \in \{0,1\}^m}{\kappa_{i_*}} \Big ([T_c]_{\alpha_{*i_*}}[\bar{T}_j]_{\alpha_{*1}}[T_j]_{\alpha_{*(1-i_*)}}-[T_j]_{\alpha_{*i_*}}[\bar{T}_c]_{\alpha_{*1}}[T_j]_{\alpha_{*(1-i_*)}} \Big).  
\]
\end{proposition}
\begin{proof}
We substitute \eqref{E-4} into the L.H.S. of \eqref{E-2} to get 
\begin{align}
\begin{aligned}\label{E-5}
\mathrm{i}\partial_t\Psi_j(t, x_{*0}) &=\mathrm{i} \sum_{\alpha_*}  \Big([\dot{T}_j]_{\alpha_*}\Phi_{\alpha_1}^1(t, x_{10})\Phi_{\alpha_2}^2(t, x_{20})\cdots\Phi_{\alpha_m}^m(t,x_{m0})  \\
& \hspace{1cm} +[{T}_j]_{\alpha_*}\partial_t(\Phi_{\alpha_1}^1(t, x_{10})\Phi_{\alpha_2}^2(t,x_{20})\cdots\Phi_{\alpha_m}^m(t, x_{m0}))\Big)\\
&=\mathrm{i}  \sum_{\alpha_*} [\dot{T}_j]_{\alpha_*}\Phi_{\alpha_1}^1(t, x_{10})\Phi_{\alpha_2}^2(t, x_{20})\cdots\Phi_{\alpha_m}^m(t,x_{m0})\\
&\hspace{0.5cm} +\sum_{\alpha_*}[{T}_j]_{\alpha_*} {\mathcal  H} \left(\Phi_{\alpha_1}^1(t, x_{10})\Phi_{\alpha_2}^2(t,x_{20})\cdots\Phi_{\alpha_m}^m(t,x_{m0})\right)\\
&=\mathrm{i}[\dot{T}_j]_{\alpha*}\Phi_{\alpha_1}^1(t,x_{10})\Phi_{\alpha_2}^2(t,x_{20})\cdots\Phi_{\alpha_m}^m(t, x_{m0})+ {\mathcal H} \Psi_j(x_{*0}).
\end{aligned}
\end{align}
Now, we equate \eqref{E-2} and \eqref{E-5} to get 
\begin{align*}
\begin{aligned}
&[\dot{T}_j]_{\alpha_1\cdots\alpha_m}\Phi_{\alpha_1}^1(t,x_{10})\Phi_{\alpha_2}^2(t, x_{20})\cdots\Phi_{\alpha_{m}}^m(t, x_{m0}) \\
&= \sum_{i_*\in\{0, 1\}^m}\kappa_{i_*}\Big(\int_{\bbt^{md}} \big(\Psi_c(x_{*i_*})\overline{\Psi_j(x_{*1})}\Psi_j(x_{*(1-i_*)})-\Psi_j(x_{*i_*})\overline{\Psi_c(x_{*1})}\Psi_j(x_{*(1-i_*)})\big)dx_{*1} \Big).
\end{aligned}
\end{align*}
This yields
\begin{align}
\begin{aligned}\label{E-7}
&[\dot{T}_j]_{\alpha_*}=\sum_{i_*\in\{0, 1\}^m}\kappa_{i_*}\int_{\bbt^{2md}} \Big(\Psi_c(x_{*i_*})\overline{\Psi_j(x_{*1})}\Psi_j(x_{*(1-i_*)})-\Psi_j(x_{*i_*})\overline{\Psi_c(x_{*1})}\Psi_j(x_{*(1-i_*)}) \Big)\\
&\hspace{3cm}\times\overline{\Phi_{\alpha_1}^1(t, x_{10})\Phi_{\alpha_2}^2(t, x_{20})\cdots\Phi_{\alpha_{m}}^m(t, x_{m0})}dx_{*1}dx_{*0}.\\
\end{aligned}
\end{align}
On the other hand, we use the relation \eqref{E-7} and  the orthogonality of $\{\Phi^1_{\alpha_{1}}(x_{10}) \otimes\cdots\otimes\Phi^m_{\alpha_{m}}(x_{m0}) \}_{\alpha_*}$ to get 
\begin{align}
\begin{aligned} \label{E-8}
&\int_{\bbt^{2md}} \Psi_c(x_{*i_*})\overline{\Psi_j(x_{*1})}\Psi_j(x_{*(1-i_*)})\overline{\Phi_{\alpha_1}^1(t,x_{10})\Phi_{\alpha_2}^2(t,x_{20} )\cdots\Phi_{\alpha_{m}}^m(t,x_{m0})}dx_{*1}dx_{*0}\\
& \hspace{5cm} =[T_c]_{\alpha_{*i_*}}[\bar{T}_j]_{\alpha_{*1}}[T_j]_{\alpha_{*(1-i_*)}}.
\end{aligned}
\end{align}
By \eqref{E-7} and \eqref{E-8}, one has the Lohe tensor model.
\end{proof}

\vspace{0.2cm}

\subsection{Emergent dynamics} \label{sec:5.3}
Suppose that $A$ and $B$ are two partitions of the set $\{1, 2, \cdots, N\}$ such that 
\[
{\mathcal I}_0 :=\{n:i_n=0,\quad 1\leq n\leq m\},\quad {\mathcal I}_1 :=\{n:i_n=1,\quad 1\leq n\leq m\}.
\]

\vspace{0.2cm}

Recall that
\begin{align*}
\begin{aligned}
&D_i := \bbt^d, \quad  D :=D_1 \times \cdots \times D_m, \quad D_A :=\prod_{n\in A}D_n,\quad D_B :=\prod_{n\in B}D_n, \\
& dx_{A0} :=\prod_{n\in A}dx_{n0}, \quad dx_{B0} :=\prod_{n\in B}dx_{n0},\quad dx_{A1}=\prod_{n\in A}dx_{n1},\quad dx_{B1} :=\prod_{n\in B}dx_{n1}.
\end{aligned}
\end{align*}
For a given configuration $\{\Psi_j \}$, we set 
\begin{equation} \label{E-9}
\Psi_c := \frac{1}{N} \sum_{j=1}^{N} \Psi_j, \quad R: = \| \Psi_c \|_2.
\end{equation}
\begin{lemma} \label{L5.1}
Let $\{\Psi_j\}$ be a global smooth solution to \eqref{E-2}. Then we have
\begin{align*}
\frac{dR^2}{dt}=\sum_{i_* \in \{0, 1\}^m}\kappa_{i_*}\sum_{i=1}^N\int_{D^2_A} \left|\int_{D_B}\big(\overline{\Psi_c(x_{*0})}\Psi_i(x_{*(1-i_*)})-\overline{\Psi_i(x_{*0})}\Psi_c(x_{*(1-i_*)})\big)dx_{B0}\right|^2dx_{A0}dx_{A1}.
\end{align*}
\end{lemma}
\begin{proof}
It follows from \eqref{E-2} and \eqref{E-9} that 
\begin{align*}
\begin{aligned}
\frac{d}{dt}\Psi_c(x_{*0}) &=-\mathrm{i} {\mathcal H} \Psi_c(x_{*0}) \\
&+\sum_{i_*}\frac{\kappa_{i_*}}{N}\int_D\sum_{i=1}^N(\Psi_c(x_{*i_*})\overline{\Psi_i(x_{*1})}\Psi_i(x_{*(1-i_*)})-\Psi_i(x_{*i_*})\overline{\Psi_c(x_{*1})}\Psi_i(x_{*(1-i_*)}))dx_{*1}.
\end{aligned}
\end{align*}
This yields
\begin{align*}
&\langle{\Psi_c| \partial_t\Psi_c}\rangle=\langle{\Psi_c| -\mathrm{i} {\mathcal H} \Psi_c}\rangle +\sum_{i_*}\frac{\kappa_{i_*}}{N}\sum_{i=1}^N\int_{D^2} \overline{\Psi_c(x_{*0})} \\
& \hspace{1.5cm} \times \Big(\Psi_c(x_{*i_*})\overline{\Psi_i(x_{*1})}\Psi_i(x_{*(1-i_*)})-\Psi_i(x_{*i_*})\overline{\Psi_c(x_{*1})}\Psi_i(x_{*(1-i_*)}) \Big)dx_{*1}dx_{*0}.
\end{align*}
Finally, one has
\begin{align*}
&\frac{d}{dt}\langle{\Psi_c|\Psi_c}\rangle=\langle{\Psi_c| \partial_t\Psi_c}\rangle+(c.c.) \\
& =\sum_{i_*}\frac{\kappa_{i_*}}{N}\sum_{i=1}^N\underbrace{\int_{D^2} \overline{\Psi_c(x_{*0})} \Big(\Psi_c(x_{*i_*})\overline{\Psi_i(x_{*1})}\Psi_i(x_{*(1-i_*)})-\Psi_i(x_{*i_*})\overline{\Psi_c(x_{*1})}\Psi_i(x_{*(1-i_*)}) \Big)dx_{*1}dx_{*0}}_{:=\mathcal{I}_{4}} \\
&+(c.c.).
\end{align*}
Now we simplify the term $\mathcal{I}_4+\overline{\mathcal{I}_4}$ as follows.
\begin{align*}
\begin{aligned}
&\mathcal{I}_4 +\overline{\mathcal{I}_4} =\int_{D^2} \Big(\overline{\Psi_c(x_{*0})}\Psi_c(x_{*i_*})\overline{\Psi_i(x_{*1})}\Psi_i(x_{*(1-i_*)})-\overline{\Psi_c(x_{*0})}\Psi_i(x_{*i_*})\overline{\Psi_c(x_{*1})}\Psi_i(x_{*(1-i_*)})  \\
& \hspace{1cm} +{\Psi_c(x_{*0})}\overline{\Psi_c(x_{*i_*})}\Psi_i(x_{*1})\overline{\Psi_i(x_{*(1-i_*)})}-{\Psi_c(x_{*0})}\overline{\Psi_i(x_{*i_*})}{\Psi_c(x_{*1})}\overline{\Psi_i(x_{*(1-i_*)})}\Big)dx_{*1}dx_{*0}.
\end{aligned}
\end{align*}
Since $x_{*0}$ and $x_{*1}$ are dummy variables, we can interchange the variables in third term and forth term in R.H.S. of above equality:
\[
x_{*0}\leftrightarrow x_{*(1-i_*)} \quad \mbox{and} \quad x_{*1}\leftrightarrow x_{*i_*}
\]
to get 
\begin{align*}
\begin{aligned}
&\mathcal{I}_4+\overline{\mathcal{I}_4} =\int_{D^2} \Big(\overline{\Psi_c(x_{*0})}\Psi_c(x_{*i_*})\overline{\Psi_i(x_{*1})}\Psi_i(x_{*(1-i_*)})-\overline{\Psi_c(x_{*0})}\Psi_i(x_{*i_*})\overline{\Psi_c(x_{*1})}\Psi_i(x_{*(1-i_*)})\\
&+{\Psi_c(x_{*(1-i_*)})}\overline{\Psi_c(x_{*1})}\Psi_i(x_{*i_*})\overline{\Psi_i(x_{*0})}-{\Psi_c(x_{*(1-i_*)})}\overline{\Psi_i(x_{*1})}{\Psi_c(x_{*i_*})}\overline{\Psi_i(x_{*0})}\big)dx_{*1}dx_{*0}\\
&=\int_{D^2} \big(\overline{\Psi_c(x_{*0})}\Psi_i(x_{*(1-i_*)})-\overline{\Psi_i(x_{*0})}\Psi_c(x_{*(1-i_*)})\big)\big(\overline{\Psi_i(x_{*1})}\Psi_c(x_{*i_*})-\overline{\Psi_c(x_{*1})}\Psi_i(x_{*i_*})\Big)dx_{*1}dx_{*0}.
\end{aligned}
\end{align*}
This yields
\begin{align*}
\begin{aligned}
&\int_{D^2} \big(\overline{\Psi_c(x_{*0})}\Psi_i(x_{*(1-i_*)})-\overline{\Psi_i(x_{*0})}\Psi_c(x_{*(1-i_*)})\big)\big(\overline{\Psi_i(x_{*1})}\Psi_c(x_{*i_*})-\overline{\Psi_c(x_{*1})}\Psi_i(x_{*i_*})\big)dx_{*1}dx_{*0}\\
& \hspace{0.2cm} =\int_{(D_A)^2} \left(\int_{D_B}\big(\overline{\Psi_c(x_{*0})}\Psi_i(x_{*(1-i_*)})-\overline{\Psi_i(x_{*0})}\Psi_c(x_{*(1-i_*)})\big)dx_{B0}\right)\\
&\hspace{1cm}\times\left(\int_{D_B}\big(\overline{\Psi_i(x_{*1})}\Psi_c(x_{*i_*})-\overline{\Psi_c(x_{*1})}\Psi_i(x_{*i_*})\big)dx_{B1}\right)dx_{A0}dx_{A1}\\
&\hspace{0.2cm} =\int_{(D_A)^2}\left|\int_{D_B}\big(\overline{\Psi_c(x_{*0})}\Psi_i(x_{*(1-i_*)})-\overline{\Psi_i(x_{*0})}\Psi_c(x_{*(1-i_*)})\big)dx_{B0}\right|^2dx_{A0}dx_{A1}.
\end{aligned}
\end{align*}
Finally, one has 
\[ \frac{dR^2}{dt} =\sum_{i_*}\frac{\kappa_{i_*}}{N}\sum_{i=1}^N\int_{(D_A)^2} \left|\int_{D_B}\big(\overline{\Psi_c(x_{*0})}\Psi_i(x_{*(1-i_*)})-\overline{\Psi_i(x_{*0})}\Psi_c(x_{*(1-i_*)})\big)dx_{B0}\right|^2dx_{A0}dx_{A1}. \]
\end{proof}
\begin{lemma}
Let $\{\Psi_j\}$ be a global smooth solution to \eqref{E-2}. Then for each $i_*$ with $\kappa_{i_*} > 0$ and $j = 1, \cdots, N$, one has 
\[ \lim_{t \to \infty} \int_{(D_A)^2} \left|\int_{D_B}\big(\overline{\Psi_c(x_{*0})}\Psi_j(x_{*(1-i_*)})-\overline{\Psi_j(x_{*0})}\Psi_c(x_{*(1-i_*)})\big)dx_{B0}\right|^2dx_{A0}dx_{A1} = 0. \]
\end{lemma}
\begin{proof}
It follows from the boundedness of $\frac{d}{dt}\Psi_i$ and the equality
\begin{align*}
\begin{aligned}
&\frac{d}{dt}\langle{\Psi_c|\Psi_c}\rangle \\
& =\sum_{i_*}\frac{\kappa_{i_*}}{N}\sum_{i=1}^N\int_{(D_A)^2} \left|\int_{D_B}\big(\overline{\Psi_c(x_{*0})}\Psi_i(x_{*(1-i_*)})-\overline{\Psi_i(x_{*0})}\Psi_c(x_{*(1-i_*)})\big)dx_{B0}\right|^2dx_{A0}dx_{A1},
\end{aligned}
\end{align*}
that we can easily obtain the boundedness of second derivative of $\Psi_i$ for all $i=1, 2, \cdots, N$. Again, by Barbalat's lemma, we have following theorem.
\end{proof}
\begin{remark}
(i)~Since each terms are nonnegative, for all $i_*\in\{0, 1\}^m$ and $i=1, 2, \cdots, N$, one has
\[
\lim_{t \to \infty} \kappa_{i_*}\int_{(D_A)^2} \left|\int_{D_B}\big(\overline{\Psi_c(x_{*0})}\Psi_i(x_{*(1-i_*)})-\overline{\Psi_i(x_{*0})}\Psi_c(x_{*(1-i_*)})\big)dx_{B0}\right|^2dx_{A0}dx_{A1} = 0.
\]
(ii)~Furthermore, if $\kappa_{i_*}>0$, for all $i_*$ and $i=1, 2, \cdots, N$, we have 
\begin{align}\label{E-14}
\lim_{t \to \infty} \int_{(D_A)^2} \left|\int_{D_B}\big(\overline{\Psi_c(x_{*0})}\Psi_i(x_{*(1-i_*)})-\overline{\Psi_i(x_{*0})}\Psi_c(x_{*(1-i_*)})\big)dx_{B0}\right|^2dx_{A0}dx_{A1} = 0.
\end{align}
\end{remark}

\vspace{0.5cm}

\begin{theorem}
Let $\{\Psi_i\}$ be a global smooth solution to \eqref{E-2} satisfying the following conditions:
\[
\kappa_{i_*}\geq0,\quad \mbox{for}~ i_* \neq (0, \cdots, 0), \qquad \kappa_{00\cdots0}>0.
\]
Then, either complete aggregation or bi-polar state occurs asymptotically.
\end{theorem}
\begin{proof}
\noindent Since $\kappa_{00\cdots0}>0$, without loss of generality, we may set
 \[ i_*\neq(0, 0, \cdots, 0) \quad \mbox{in \eqref{E-14}.} \]
 Then, we have the term involving with $\kappa_{00\cdots0}$:
\begin{align*}
&\int_{D^2}\left|\big(\overline{\Psi_c(x_*)}\Psi_i(y_*)-\overline{\Psi_i(x_*)}\Psi_c(y_*)\big)\right|^2dx_*dy_*\\
& \hspace{1cm} =\int_{D^2}\big(\overline{\Psi_c(x_*)}\Psi_i(y_*)-\overline{\Psi_i(x_*)}\Psi_c(y_*)\big)\big({\Psi_c(x_*)}\overline{\Psi_i(y_*)}-{\Psi_i(x_*)}\overline{\Psi_c(y_*)}\big)dx_*dy_*\\
&\hspace{1cm} =2\|\Psi_c\|^2_2\cdot\|\Psi_i\|_2^2-\langle{\Psi_c, \Psi_i}\rangle_F^2-\langle{\Psi_i, \Psi_c}\rangle_F^2\\
&\hspace{1cm} =2\left(\|\Psi_c\|^2_2\cdot\|\Psi_i\|_2^2-|\langle{\Psi_c, \Psi_i}\rangle|^2\right)+4\mathrm{Im}(\langle{\Psi_c, \Psi_i}\rangle^2).
\end{align*}
This yields
\[ \lim_{t \to \infty} \Big( \|\Psi_c\|^2_2\cdot\|\Psi_i\|_2^2-|\langle{\Psi_c, \Psi_i}\rangle|^2 \Big) = 0,\quad \lim_{t \to \infty} \mathrm{Im}(\langle{\Psi_c, \Psi_i}\rangle^2) =0.
\]
So we have
\begin{align*}
\begin{aligned}
\|\Psi_c-\langle{\Psi_c, \Psi_i}\rangle \Psi_i\|^2_2&=\|\Psi_c\|_2^2+|\langle{\Psi_c, \Psi_i}\rangle|^2-\langle{\Psi_c, \Psi_i}\rangle^2-\langle{\Psi_i, \Psi_c}\rangle^2\\
&=\|\Psi_c\|^2_2\cdot\|\Psi_i\|_2^2-|\langle{\Psi_c, \Psi_i}\rangle|^2+2\mathrm{Im}(\langle{\Psi_c, \Psi_i}\rangle^2)\rightarrow 0,
\end{aligned}
\end{align*}
as time goes to infinity. Hence, we know that there exists complex scalar function $\lambda_i(t)$ such that 
\begin{align}\label{E-15}
\|\Psi_c(t)-\lambda_i(t)\Psi_i(t)\|_2^2\rightarrow0.
\end{align}
So we know
\begin{align*}
\|\Psi_c-\langle{\Psi_c, \Psi_i}\rangle \Psi_i\|_2&=\|(\Psi_c-\lambda_i\Psi_i)+\lambda_i\Psi_i-\langle{(\Psi_c-\lambda_i\Psi_i)+\lambda_i\Psi_i, \Psi_i}\rangle\Psi_i\|_2\\
&\geq\big| \|\Psi_c-\lambda_i\Psi_i-\langle \Psi_c-\lambda_i\Psi_i, \Psi_i\rangle \Psi_i\|_2-\|\lambda_i\Psi_i-\bar{\lambda}_i\Psi_i\|_2\big|.
\end{align*}
On the other hand, it follows from \eqref{E-15} that
\[ \lim_{t \to \infty} \|\lambda_i \Psi_i-\bar{\lambda_i}\Psi_i\|_2 = 0.
\]
That means we can set $\lambda_i$ be real number. Also from $\|\Psi_c\|_2=|\lambda_i|$ and $\|\Psi_c\|_2$ is nondecreasing, we can set
\[
|\lambda_i(t)|\geq \|\Psi_c(0)\|_2.
\]
From the triangle inequality
\[
\|\lambda_1\Psi_1-\lambda_i\Psi_i\|_2=\Psi_1\|(\Psi_c-\lambda_i\Psi_i)-(\Psi_c-\lambda_1\Psi_1)\|_2\leq \|\Psi_c-\lambda_i\Psi_i\|_2+\|\Psi_c-\lambda_1\Psi_1\|_2.
\]
If we set $a_i=\lambda_i/\lambda_1$, we have
\[
\Psi_i-a_i\Psi_1\rightarrow0.
\]
Since $|a_i|=1$ and $a_i$ are real numbers,
\[
a_i=\pm1.
\]
This completes the proof.
\end{proof}
\begin{lemma} \label{L5.3}
Let $\{ \Psi_i \}$ be an ensemble with
\[
\Psi_1=\cdots=\Psi_n=\Psi^\infty, \quad \Psi_{n+1}=\cdots=\Psi_N=-\Psi^{\infty},
\]
where $0\leq n\leq N/2$ and $\|\Psi^\infty\|_2=1$. Then we have
\[
\|\Psi_c\|_2=\left|1-\frac{2n}{N}\right|.
\]
\end{lemma}
\begin{proof} By direct calculation, one has
\[
\Psi_c=\left(\frac{N-2n}{N}\right)\Psi^{\infty}.
\]
Thus, we have
\[
\|\Psi_c\|_2=1-\frac{2n}{N}
\]
\end{proof}
\begin{remark}
If $\|\Psi_c(0)\|_F>1-\frac{2}{N}$, then a bi-polar state is impossible.
\end{remark}

From above remark, we have following theorem.

\begin{theorem}
Suppose that the Hamiltonian, coupling strengths and initial data satisfy
\[ {\mathcal H} = 0, \quad \kappa_{i_*}\geq0,\quad \forall i_*\neq(0, 0, \cdots, 0), \quad \kappa_{00\cdots0}>0,\qquad \|\Psi_c(0)\|_2 >1-\frac{2}{N}, \]
and let $\{\Psi_i\}$ be a global smooth solution to \eqref{E-2}. Then the complete aggregation occurs asymptotically.
\end{theorem}
\begin{proof}
We use \eqref{L5.1} and assumption to see 
\begin{equation} \label{D-22}
 R(t) \geq R(0) >  1-\frac{2}{N}, \quad t \geq 0. 
\end{equation} 
Suppose that a bi-polar state emerges: for some $n \leq [N/2]$, one has
\begin{align*}
\begin{aligned} 
& \lim_{t \to \infty} \| \Psi_j(t) -\Psi \|_2 = 0, \quad  1 \leq j \leq n, \\
&  \lim_{t \to \infty} \| \Psi_j(t) -(-\Psi) \|_2 = 0, \quad  n+1 \leq j \leq N.
\end{aligned}
\end{align*}
Then, it follows from Lemma \ref{L5.3} that 
\[  \lim_{t \to \infty} R(t) =1-\frac{2n}{N},   \]
which is clearly contradictory to \eqref{D-22}.  Hence, we have the complete state aggregation.
\end{proof}

\section{Conclusion} \label{sec:6}
\setcounter{equation}{0}
In this paper, we have proposed an infinite-dimensional Schr\"{o}dinger-Lohe hierarchy consisting of the Schr\"{o}dinger-Lohe model, the Schr\"{o}dinger-Lohe matrix model and Schr\"{o}dinger-Lohe tensor model. In a series of recent papers, the authors established the Lohe hierarchy consisting of the Kuramtoo model, the Lohe sphere model, the Lohe matrix model and the Lohe tensor model. Prior to this work, the relation between the Schr\"{o}dinger-Lohe model and the Lohe matrix model was a kind of mystery that remained unsolved in last ten years.

In this work, we have shown that the infinite-dimensioal analog of the complex Lohe sphere model appears as a coefficient system of the  Schr\"{o}dinger-Lohe model. Thanks to this explicit connection between the complex Lohe sphere model and the  Schr\"{o}dinger-Lohe  model, we establish an infinite-dimensional Schr\"{o}dinger-Lohe hierarchy (see the diagram below): 
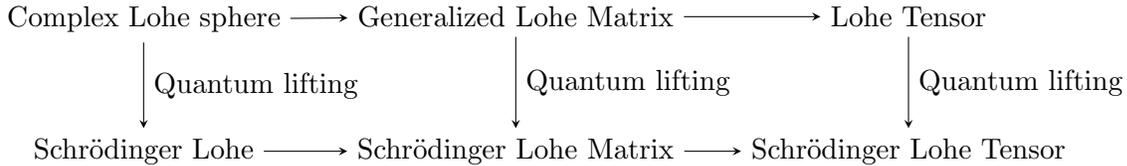
\begin{center}
\begin{tikzpicture}
  \matrix (m) [matrix of math nodes,row sep=3em,column sep=2em,minimum width=2em]
  {
     \mbox{Complex Lohe sphere} & \mbox{Generalized Lohe Matrix} &\mbox{Lohe Tensor}\\
     \mbox{Schr\"{o}dinger Lohe} & \mbox{Schr\"{o}dinger Lohe Matrix}&\mbox{Schr\"{o}dinger Lohe Tensor} \\};
  \path[-stealth]
    (m-1-1) edge node [right] {Quantum lifting} (m-2-1)
            edge [right] node [below] {} (m-1-2)
    (m-2-1) edge node [below] {}
            node [above] {} (m-2-2)
    (m-1-2) edge node [right] {Quantum lifting} (m-2-2)
    (m-1-2) edge node [below] {} (m-1-3)
    (m-2-2) edge node [below] {} (m-2-3)
    (m-1-3) edge node [right] {Quantum lifting} (m-2-3);
\end{tikzpicture}
\end{center}
There are many unresolved issues related to this work. For example, in this paper, we considered the homogeneous ensemble with the same free flow. Thus, analysis on the emergent dynamics of heterogeneous ensemble is still far from complete, for example, we do not have a good analysis on the complete aggregation of the Schr\"{o}dinger-Lohe model except a weak result on the practical aggregation. These issues will be left for a future work.

%

\end{document}